\documentclass{article}[12pt]
\usepackage{latexsym,color}
\usepackage{graphicx}
\usepackage{mathptmx}
\usepackage{amssymb}
\usepackage{amsfonts}
\usepackage{amsmath}
\usepackage{latexsym}

\numberwithin{equation}{section}
\newtheorem{thm}{Theorem}[section]
\newtheorem{dfn}[thm]{Definition}
\newtheorem{lem}[thm]{Lemma}
\newtheorem{rem}[thm]{Remark}
\newtheorem{cor}[thm]{Corollary}

\newtheorem{asm}[thm]{Assumption}
\newtheorem{example}[thm]{Example}

\newenvironment{proof}{\noindent {\bf Proof.\/}}{$\qed$\vskip 0.1in}
\def\qed{ \hfill \vrule width.2cm height.2cm depth0cm\smallskip}

\begin{document}

\title{Duality and Convergence for
 Binomial Markets with Friction
 \thanks{Research supported by the
European Research Council Grant 228053-FiRM,
the Swiss Finance Institute and the ETH Foundation.
The authors would like to thank Prof. Kusuoka and Marcel Nutz
for insightful discussions.}
}

\date{June 9, 2011}
\author{
  Yan Dolinsky
  \thanks{
  ETH Zurich, Dept.\ of Mathematics,
  \texttt{yan.dolinsky@math.ethz.ch}
  }
  \and
  H.\ Mete Soner
  \thanks{
  ETH Zurich, Dept.\ of Mathematics,
  and Swiss Finance Institute, \texttt{mete.soner@math.ethz.ch}}
  }

\maketitle

\begin{abstract}
We prove limit theorems for the super-replication
cost of European options in a Binomial model
with friction.  The examples covered are markets
with proportional transaction costs and
the illiquid markets.
The dual representation for the super-replication
cost in these models are obtained and used to prove the limit theorems.
In particular, the existence of the liquidity premium
for the continuous time limit of the
model proposed in \cite{CJP} is proved.
Hence, this paper extends the
previous convergence result of
\cite{GS} to the general non-Markovian case.
Moreover, the special
case of small transaction
costs yields, in the continuous limit, the $G$-expectation
of Peng
as earlier proved by Kusuoka in \cite{K}.
\end{abstract}

{\small



\section{Introduction}
\label{intro}
We consider a one-dimensional Binomial model
in which the size of the trade has an immediate but temporary
effect on the price of the asset. Indeed, let
$g(t,\nu)$ be the cost of trading $\nu$
shares at time $t$.  We simply assume that $g$
is adapted to the natural filtration
and it is convex in $\nu$ with $g(t,0)=0$.
In this generality this model corresponds to the
classical transaction cost model when
$g(t,\nu) = \lambda |\nu|$
with a given constant $\lambda >0$.  However,
it also covers the illiquidity model
considered in \cite{cer} and \cite{GS} which is
the Binomial version of the model introduced
by Cetin, Jarrow and Protter in \cite{CJP} for continuous time.
In this example, $g$ is twice differentiable at $\nu=0$.

In continuous time the super-replication cost of a European
option behaves quite differently depending on the structure of $g$.
In the case of proportional transaction costs (i.e. when
$g$ is non-differentiable at the origin),
the super-replication cost is very costly
as proved in \cite{SSC,LS,D1}.
In several papers \cite{BS,K} asymptotic
problems with vanishing transaction costs are
considered to obtain non-trivial pricing
equations. On the other hand,
when $g$ is differentiable then any continuous trading
strategy which has finite variation, has no liquidity cost.
Thus one may avoid the liquidity cost entirely
as shown in \cite{CJP} and also in \cite{BB}.
However, in \cite{CST} it
is shown that mild constraints on the admissible
strategies render these approximation inadmissible
and one has a liquidity premium.  This result
is further verified in \cite{GS} which derives
the same premium as the continuous
time limit of Binomial models. The
equation satisfied by this limit is
a nonlinear Black-Scholes equation
\begin{equation}
\label{e.bs}
-u_t(t,s)+ \frac{\sigma^2 s^2}{2} H(u_{ss}(t,s)) =0,
\quad
\forall\
t<T, s>0,
\end{equation}
where $t$ is time, $s$ is the current stock price
and  $H$ is a convex nonlinear function of the
second derivative
derived explicitly in \cite{CST,GS}. Since $H$ is convex,
the above equation
is the dynamic programming equation of a stochastic optimal
control problem.  Then this problem may be considered
as the dual of the original super-replication problem.

The proof
given in \cite{GS} depends on the homogenization
techniques for viscosity solutions.  Thus
it is limited to the Markovian claims.  Moreover,
the mentioned duality result is obtained
only through the partial differential equation
and not by a direct argument.

In this paper we extend the study of \cite{GS}
to non Markovian claims and to more general
liquidity functions $g$.
The model is again a  simple one dimensional
model with trading cost $g$.
In this formulation, the super-replication problem
is a convex optimization problem
and its dual can be derived by  the classical
theory.
This derivation is an advantage of the discrete model
as the derivation of the dual in continuous time
is essentially an open problem.  Although
 a new approach is now developed in \cite{STZ}.
The dual is an optimal control problem
in which the controller is allowed to choose different
probability measures.
We then use this dual representation to formally
identify the limit optimal control problem.
The dynamic programming equation
of this optimal control problem is
given by (\ref{e.bs}) in the Markov case.
This  representation also allows us to prove
 the continuous time limit.

Our approach is purely probabilistic and allows us to deal
with path dependent payoffs and path dependent
penalty functions $g$.
One of the key step is a construction of Kusuoka
given in the context of transaction costs.
Indeed, given a martingale $M$ on the Brownian probability
space whose volatility satisfies some regularity conditions,
Kusuoka in \cite{K} constructs
a sequence of martingales on the discrete probability
space $\{-1,1\}^\infty$ of a specific form
 which converge in law  to $M$.
 Moreover, the quadratic variation of $M$ is
 approximated through this powerful procedure
 of Kusuoka.
 This construction is
our main tool in proving the lower bound (i.e.,
existence of liquidity premium)
for the continuous time limit of the super-replication costs.
The upper bound follows from
compactness and two general lemmas (Lemmas \ref{convergence}
and \ref{randomization}).

As remarked before, the super-replication cost
can be quite costly in markets with transaction costs.  Therefore
if $g(t,\nu)= \lambda |\nu|$ and $\lambda>0$ is a given constant,
one obtains a trivial result in the continuous time limit.  So
we need to scale the proportionality constant $\lambda$
as the time discretization gets smaller.  Indeed, if in an $n$-step
model, we take $\lambda_n = \Lambda/\sqrt{n}$
then the limit problem is the uncertain volatility model
or equivalently $G$-expectation of Peng \cite{P}.
This is exactly the main result of Kusuoka in \cite{K}.
In fact, relatedly,
the authors in joint work with M. Nutz \cite{DNS}
provides a different discretization of
the $G$-expectation.

The paper is organized as following.
In the next section we introduce the setup.
In Section 3 we formulate the main results of this paper.
In Section 4 we prove Theorem \ref{thm2.1}, that is
a duality result for the super--replication prices
in the binomial models. The main tool that is used in this section
is the Kuhn-Tucker theory for convex optimization.
Theorem \ref{thm2.2}
which describes the
asymptotic behavior of the super--replication prices,
is proved in Section 5.
In Section 6 we state the main results from \cite{K},
which are used in this paper. In particular we give a short formulation of the main properties of
Kusuoka construction, which is the main tool in proving the lower bound (liquidity premium) of Theorem \ref{thm2.2}.
In Section 7 we derive auxiliary lemmas, Lemmas \ref{convergence}--\ref{randomization}
and  Lemma \ref{density} that are used
in the proof of Theorem \ref{thm2.2}.

\section{Preliminaries and the model}
\label{sec:2}
\setcounter{equation}{0}
Let $\Omega={\{-1,1\}}^\infty$ be the space of infinite sequences
$\omega=(\omega_1,\omega_2,...)$;
$\omega_i\in\{-1,1\}$ with the product probability
$\mathbb{Q}=\{\frac{1}{2},\frac{1}{2}\}^\infty$.
Define the canonical sequence of i.i.d.~random variables $\xi_1,\xi_2,...$ by
$$
\xi_i(\omega)=\omega_i, \qquad {i}\in\mathbb{N},
$$
and consider the natural filtration
$\mathcal{F}_k=\sigma{\{\xi_1,...,\xi_k}\}$, $k\geq 1$
and let $\mathcal{F}_0$ be trivial.

For any ${T}>0$ denote by $\mathcal{C}[0,{T}]$ the
space of all continuous functions on $[0,{T}]$
with the
uniform topology induced by the norm
 $\|y\|_\infty=\sup_{0\leq t\leq {T}}|y(t)|$.
Let
$F:\mathcal{C}[0,T]\rightarrow\mathbb{R}_{+}$
be a continuous map such that there are
constant $C, p>0$ for which
\begin{equation}\label{2.2}
\begin{split}
F(y)\leq C(1+\|y\|_\infty^p), \
\ \ \forall{y}\in \mathcal{C}[0,T].
\end{split}
\end{equation}
Without loss of generality we take $T=1$.

Next, we introduce a sequence of binomial models for which the
volatility of the stock price
is a constant $\sigma>0$ (which is independent of $n$). Namely,
for any $n$ consider the $n$--step
binomial model of a financial market which is active at times
$0,1/n,2/n,\ldots,1$. It consists of a savings account,
and of a stock. Without loss of generality (by discounting), we assume
that the savings account price is a constant which equals to 1. The stock price at time
$k/n$ is given by
\begin{equation}
\label{2.4}
S^{(n)}(k)=s_0\exp\left(\sigma\sqrt\frac{1}{n}\sum_{i=1}^{k} \xi_i\right),
 \ \ \ k=0,1,...,n
\end{equation}
where $s_0>0$ is the initial stock price.
For any $n\in\mathbb{N}$, let
$\mathcal{W}_n:\mathbb{R}^{n+1}\rightarrow\mathcal{C}[0,1]$
be the linear interpolation operator given by
$$
\mathcal{W}_{n}(y)(t):=
\left(\left[{nt}\right]+1-{nt}\right)y\left(\left[{nt}\right]\right)+
\left(nt-\left[nt\right]\right)y\left(\left[nt\right]+1\right),
 \ \ \forall{t}\in[0,1]
$$
where $y=\{y(k)\}_{k=0}^n\in\mathbb{R}^{n+1}$ and $[z]$ denotes the integer part of $z$.
We consider a (path dependent) European contingent
claim with maturity  $T=1$ and
a payoff given by
\begin{equation}\label{2.5}
F_n:=F\big(\mathcal{W}_n(S^{(n)})\big)
\end{equation}
where, by definition we consider, $\mathcal{W}_n(S^{(n)})$ as a random
element in $\mathcal{C}[0,1]$.

For future reference,  Let $\mathcal{C}^{+}[0,1]$ be the set of all strictly positive
continuous functions on $[0,1]$ with the uniform topology.
Then, in fact $\mathcal{W}_n(S^{(n)})$ is an element in
$\mathcal{C}^{+}[0,1]$

\subsection{Wealth dynamics and super-replication}

Next, we define the notion of a self financing portfolio
in these models. Fix $n\in\mathbb{N}$
and consider an $n$-step binomial model,
with a penalty function $g$.
We assume that this function represents
the cost of trading in this market.  We assume the following.

\begin{asm}
\label{a.g}
{\rm The} trading cost function
$$
g: [0,1] \times \mathcal{C}^+[0,1] \times \mathbb{R}
\to [0,\infty)
$$
{\rm
is assumed to be non-negative, adapted with $g(t,S,0)=0$.
Moreover we assume that $g(t,S,\cdot)$ is convex for every
$(t,S) \in [0,1]\times  \mathcal{C}[0,1]$.}
\end{asm}

In this simple setting, the adaptedness of $g$
simply means that $g(t,S,\nu)$ depends only on the restriction of
$S$ to the interval $[0,t]$, namely
$$
g(t,S,\nu)=g(t,\hat{S},\nu)\quad
{\mbox{whenever}}\quad
S(s)=\hat{S}(s)\ \forall \ s\leq t.
$$

A self financing
portfolio $\pi$ with an initial capital $x$
is a pair $\pi=(x,\{\gamma(k)\}_{k=0}^{n})$ where
$\gamma(0)=0$ and for any $k\geq 1$,
$\gamma(k)$ is a $\mathcal{F}_{k-1}$ measurable
random variable.  Here $\gamma(k)$ represents
the number of stocks that the investor holds
at the moment $(k/n)$, {\em before}
 a transfer is made at this time.
The portfolio value $ Y^\pi(k):= Y^\pi(k: g)$, of a
trading strategy $\pi$ is given by the
difference equation
\begin{equation}
\label{2.7}
Y^{\pi}(k+1)= Y^\pi(k)+\gamma(k+1)\left(S^{(n)}(k+1)-S^{(n)}(k)\right) -
g\left(\frac{k}{n},\mathcal{W}_n(S^{(n)}),\gamma(k+1)-\gamma(k)\right),
\end{equation}
for $k=0,\ldots, n-1$ and
with initial data
$Y^{\pi}(0)=x$.

Observe that $Y^\pi(k)$ is the portfolio value
at the time $(k/n)$ {\em before} a transfer is made at this time,
and the last term in equation (\ref{2.7}) represents
the cost of trading and it is the only source of friction
in the model.
We would mostly use the notation  $Y^\pi(k)$
when the dependence on the penalty function
is clear.

Let $\mathcal{A}_n(x)$ be the set of all portfolios
with an initial capital $x$.
The problem we consider is the
super-replication cost
of a European claim whose pay-off is given in (\ref{2.5}).
Then, the problem is
\begin{equation}
\label{e.def}
V_n:=V_n(g, F_n)=\inf\left\{x\ |\ \exists\ \pi\in\mathcal{A}_n(x) \ \
\mbox{such} \ \ \mbox{that}\  Y^\pi(n: g)\geq F_n,\
 \ \mathbb{Q}\mbox{-a.s.}
\right\}.
\end{equation}

\subsection{Trading cost}
\label{ss.penalty}

In this subsection, we state the main
assumption on $g$ in addition to Assumption \ref{a.g}.
We also provide several examples and
make the connection to the models with proportional
transaction costs
and models with price impact.

Let
$G:[0,1]\times\mathcal{C}^{+}[0,1]
\times\mathbb{R}\rightarrow [0,\infty]$,
be the Legendre transform (or convex conjugate) of $g$,
\begin{equation}\label{2.6}
\begin{split}
G(t,S,y)=\sup_{\nu\in\mathbb{R}}(\nu y-g(t,S,\nu)),
\ \ \forall{\ (t,S,y)} \in [0,1]\times\mathcal{C}^{+}[0,1]
\times\mathbb{R}.
\end{split}
\end{equation}
Observe that $G$ may become infinite.
It is well known that
the following dual relation holds,
$$
g(t,S,\nu)=\sup_{y\in\mathbb{R}}(\nu y-G(t,S,y)),
\ \ \forall{\ (t,S,\nu)} \in [0,1]\times\mathcal{C}^{+}[0,1]
\times\mathbb{R}.
$$
\begin{example}
\label{example1}
{\rm The following three cases
provide the essential examples
of the theory developed in this paper.

\noindent
{\bf{a.}}  For a given constant $\Lambda >0$, let
$$
g(t,S,\nu) = \Lambda \nu^2.
$$
In this example , we directly calculate that
$$
G(t,S,y) = y^2/(4\Lambda).
$$
This penalty function is
the Binomial version of the linear liquidity model of
Cetin, Jarrow, Protter \cite{CJP} that was studied in \cite{GS}
(see Remark \ref{r.liquidity} below).

In \cite{st}, it is proved that the optimal trading strategies
in continuous time do not have jumps.  Hence one expects that
in a Binomial model with large $n$,
the optimal portfolio changes are also small.
Thus any trading cost $g$ which is twice differentiable
essentially behaves like this example with $\Lambda = g_{\nu \nu}(t,S,0)$.

\noindent
{\bf{b.}}  This example which corresponds to
the example of proportional transaction costs.
For fix $n$ recently there has been
interesting results in relation to arbitrage. We refer
to the paper of Schachermayer \cite{Sc},
Pennanen and Penner \cite{PP} and the references
therein.
But as remarked earlier, fixed transaction cost
forces the super-replication to be very costly as $n$
tends to infinity. Hence
we take a sequence of problems with vanishing
transaction costs,
$$
g^c_n(t,S,\nu) = \frac{c}{\sqrt{n}} \ S(t) |\nu|,
$$
where $c>0$ is a constant. This discrete financial market with
vanishing transaction costs is exactly
the  model studied in \cite{K} by Kusuoka.
In this case, the dual function is given by
$$
G^c_n(t,S,y) =
\left\{
\begin{array}{ll}
0,\qquad &{\mbox{if}}\ |y| \le {c\ S(t)}{/\sqrt{n}},\\
+\infty, &{\mbox{else.}}
\end{array}
\right.
$$

\noindent
{\bf{c.}}  This example is a mixture of the previous
two.  It is obtained by appropriately
modifying the liquidity example.  In our analysis
this modification will be used in several places.
For a given constant $c$, let
$$
G^c_n(t,S,y) =
\left\{
\begin{array}{ll}
y^2/4 \Lambda,\qquad &{\mbox{if}}\ |y| \le {c S(t)}{/\sqrt{n}},\\
+\infty, &{\mbox{else.}}
\end{array}
\right.
$$
We directly calculate that
$$
g^c_n(t,S,\nu) =
\left\{
\begin{array}{ll}
\Lambda \nu^2 ,\qquad &{\mbox{if}}\ |\nu| \le
\frac{c S(t)}{2\sqrt{n} \Lambda},\\
\frac{c }{\sqrt{n}} S(t) |\nu| -
\frac{c^2 S^2(t)}{4 n \Lambda}, &{\mbox{else.}}
\end{array}
\right.
$$

\qed}
\end{example}

In the above, the third example is obtained from
the first one through an appropriate truncation
of the dual cost function $G$.
One may perform the same modification
to all given penalty functions $g$.  The following
definition formalizes this.

\begin{dfn}
\label{d.modify}
{\rm Let $g:[0,1]\times\mathcal{C}^{+}[0,1]
\times\mathbb{R}\rightarrow [0, \infty]$ be a
convex function with $g(t,S,0)=0$.
Then the} truncation of $g$ at
level $c$ {\rm is given by
$$
g_n^c(t,S,\nu):= g_n^c(t,S,\nu:g) =
\sup\left(\nu y-G(t,s,y)\ |\ |y| \le cS(t)/\sqrt{n}\ \right),
$$
where $G$ is the convex conjugate of the original $g$.}
\end{dfn}

An important but a simple observation
is the structure of the dual function of $g_n^c$.
Indeed, it is clear that the Legendre transform $G_n^c$
of $g_n^c$ is simply given by

\begin{equation}
\label{e.Gnc}
G_n^c(t,S,y) =
\left\{
\begin{array}{ll}
G(t,S,y),\qquad &{\mbox{if}}\  \ |y| \le cS(t)/\sqrt{n},\\
+\infty, &{\mbox{else.}}
\end{array}
\right.,
\end{equation}
where $G$ is the Legendre transform of $g$.

Note that for any $n\in\mathbb{N}$, $g_n^c$ converges
monotonically to $g$
as $c$ tends to infinity. Also observe
that Example 2.1.b is the truncation of the
following function
$$
g(t,S,\nu)=
\left\{
\begin{array}{ll}
0,\qquad &{\mbox{if}}\ \nu=0,\\
+\infty, &{\mbox{else.}}
\end{array}
\right.
$$
 Example 2.1.c, however, corresponds to
 the truncation of $g(t,S,\nu)=\Lambda\nu^2$.

We close this subsection, by connecting
the above model to the discrete liquidity models.

\begin{rem}
\label{r.liquidity}
{\rm
Following the liquidity model which was introduced in \cite{CJP},
we introduce a path dependent supply curve,
$$
\mathbf{S}:[0,1]\times\mathcal{C}^{+}[0,1]
\times\mathbb{R}\rightarrow\mathbb{R}.
$$
We assume  that $\mathbf{S}(t,S,\cdot)$ is adapted, i.e.,
it depends only on the restriction of $S$ to the interval $[0,t]$, namely
$$
\mathbf{S}(t,S,\nu)=\mathbf{S}(t,\hat{S},\nu)\quad
{\mbox{whenever}}\quad
S(s)=\hat{S}(s)\ \forall \ s\leq t.
$$
In the $n$-step binomial model, the price per stock
share at time $t$ is given by $\mathbf{S}(t,
\mathcal{W}_n(S^{(n)}),\nu)$,
where $\nu$ is the size of the transactions of the investor.
The penalty which represents the liquidity effect of the model is
then given by
$$
g(t,S,\nu)= \left(\ \mathbf{S}(t,S,\nu)-S(t)\right)\ \nu, \ \ \forall{\ (t,S,\nu)}
\in [0,1]\times\mathcal{C}^{+}[0,1]
\times\mathbb{R}.
$$
\qed}
\end{rem}

\section{Main results}
\label{sec:2+}
\setcounter{equation}{0}

Our first result is the characterization of
the dual problem.
We believe that this simple result
is quite interesting by itself.  Also
it will be the essential tool
to study the asymptotic behavior of the
super-replication costs.

Recall that $F_n$ and $V_n$ are
given, respectively,
in  (\ref{2.5}) and (\ref{e.def}).  Moreover, $g$ is
the trading cost function and $G$ is
its Legendre transform.

\begin{thm}[Duality]
\label{thm2.1}
Let $\mathcal{Q}_n$ be the set of all
probability measures on $(\Omega,\mathcal{F}_n)$.
Then
$$
V_n=\sup_{\mathbb{P}\in \mathcal{Q}_n}\mathbb{E}^\mathbb{P}
\left[F_n-\sum_{k=0}^{n-1}
G\left(\frac{k}{n}, \mathcal{W}_n(S^{(n)}),\mathbb{E}^\mathbb{P}\left[S^{(n)}(n)
 \left|  \right.\mathcal{F}_k\right]-
S^{(n)}(k)\right)\right],
$$
where $\mathbb{E}^\mathbb{P}$ denotes the expectation with respect
to the probability measure $\mathbb{P}$.
\end{thm}

The duality is proved in the next section.

In the limit theorem that we state below,
we assume that
the  Legendre transform $G$ of the
 convex penalty function $g$
satisfies the following.

\begin{asm}\label{a.limit}
We assume that $G$ satisfies the following growth and scaling conditions.

a). There are constants $C,p>0$ and $\beta \ge 2$
such that
\begin{equation}
\label{2.13}
{G}(t,S,y)\leq C\left|y\right|^\beta(1+\|S\|_\infty)^p, \ \ \forall(t,S,y)\in [0,1]\times \mathcal{C}^{+}[0,1]\times\mathbb{R}.
\end{equation}

b). There exists a continuous function
$$
\widehat{G}:[0,1]\times\mathcal{C}^{+}[0,1]\times\mathbb{R}\rightarrow
[0,\infty),
$$
such that
for any bounded sequence ${\{\alpha_n\}}$, discrete valued sequence $\xi_n \in \{-1,1\}$
and convergent sequences $t_n \to t$,
$S^{(n)} \to S$ (in the $\|\cdot \|_\infty$-norm),
\begin{equation}\label{2.13+}
\lim_{n\rightarrow\infty}\left|n G\left(t_n,S^{(n)},\frac{\xi_n \alpha_n}{\sqrt n}S^{(n)}(t_n)\right)
- \widehat{G}\left(t,S,\alpha_nS(t)\right)\right|=0.
\end{equation}
\end{asm}

It is straightforward to show that $\widehat{G}$ is quadratic in the $y$-variable.
Moreover,
the above assumption (\ref{2.13+})
is essentially equivalent to assume that $G$ is twice differentiable at the origin.
Indeed, when $G$ twice differentiable, Taylor approximation implies that
$$
\widehat{G}(t,S,y) = \frac12 y^2 G_{yy}(t,S,0).
$$

 We give the following example to clarify the above assumption.

 \begin{example}
\label{e.powers}
{\rm{{For $\gamma \ge 1$, let
$$
g_\gamma(\nu)=\frac{1}\gamma\  |\nu|^{\gamma}.
$$
Then, for $\gamma >1$
$$
G_\gamma(y)= \frac{1}{\gamma^*}\ |y|^{\gamma^*}, \ \ \gamma^*= \frac{\gamma}{\gamma-1}.
$$
For $\gamma=1$, $G_1(y)=0$ for
$|y| \le 1$ and is equal to infinity
otherwise.
Moreover,  we directly calculate that $\widehat{G}_\gamma(0)=0$
and for $y \neq 0$,
$$
\widehat{G}_\gamma(y):= \lim_{n \to \infty}\ n G_\gamma\left(
\frac{y}{\sqrt{n}}\right)
=
\left\{
\begin{array}{ll}
G_2(t,y),\ \ & {\mbox{if}}\ \gamma =2,\\
0, & {\mbox{if}}\ \gamma \in [1,2), \\
+\infty, & {\mbox{if}}\ \gamma >2.
\end{array}\right. .
$$
Notice that $G_\gamma$ is twice differentiable at the origin only for $\gamma \in [1,2]$.
\qed}}}

\end{example}

To describe the continuous time limit,
we need to introduce some further notation.
Let
$(\Omega_W, \mathcal{F}^{W}, \mathbb{P}^{W})$
be a complete probability space together with a standard
one-dimensional  Brownian motion
$W$ and the right continuous filtration
$\mathcal{F}^{W}_t=\sigma\big\{\sigma{\{W(s)|s\leq{t}\}
\bigcup\mathcal{N}}\big\}$,
where $\mathcal{N}$
is the collection of all $\mathbb{P}^W$ null sets.
For any $\alpha$
progressively measurable, bounded, real-valued
process, let $S_{\alpha}(t)$
be the continuous martingale given by
\begin{equation}\label{2.14}
\begin{split}
S_{\alpha}(t)=s_0\exp\left(\int_{0}^t \alpha(u) dW(u)-
\frac{1}{2}\int_{0}^t \alpha^2(u) du \right),
\ \ t\in [0,1].
\end{split}
\end{equation}
We also introduce the following
notation which is related to
the quadratic variation density
of $\ln S_\alpha$.
Recall that the constant $\sigma$ is
the volatility that was already introduced in the
dynamics of the discrete stock price process in
{\rm{(}}\ref{2.4}{\rm{)}}.
\begin{equation}
\label{e.a}
a(t:S_\alpha):=\frac{\frac{d\langle\ln S_\alpha\rangle(t)}{dt}- \sigma^2}{2\sigma}
= \frac{\alpha^2(t)-\sigma^2}{2\sigma}.
\end{equation}
The continuous limit is given
through an optimal control problem
in which  $\alpha$ is the control
and $S_\alpha$ is the controlled state
process.  To complete description of this
control problem, we need to specify the
set of admissible controls.

\begin{dfn}
\label{d.admissible}
{\rm
For any constant $c>0$,} an admissible
control at the level $c$
{\rm is a progressively measurable,
real-valued process $\alpha(\cdot)$
satisfying
$$
\left|a(\cdot:S_\alpha)\right|\le c, \ \
{\cal{L}}\otimes \mathbb{P}^W
 \ \ \mbox{a.s.},
 $$
 where ${\cal{L}}$ is the
 Lebesgue measure on $[0,1]$.}
The set of all admissible controls
{\rm is denoted by ${\cal{A}}^c$.}
\end{dfn}

As before
$g$ is the penalty function and
$g_n^c$ is the truncation of $g$ at the level $c$ as defined
in Definition \ref{d.modify}.
Let $F_n$ be a given claim and
$V_n=V_n(g,F_n)$ be the super-replication  cost defined in (\ref{e.def}).
For any level $c$, let $V^c_n=V_n(g_n^c,F_n)$.

The following theorem, which will be proved in Section \ref{sec5},
is the main result of the paper. It provides
the asymptotic behavior of the
truncated super-replication costs $V_n^c$.
Since $V_n^c \leq V_n$
for every $c$,
the below result can be used to show
the existence of a liquidity premium
as it was done for a Markovian example
in \cite{GS}, see Corollary \ref{cor2.1} and Remark \ref{rem.liquidity} below.

\begin{thm}[Convergence]
\label{thm2.2}
Let $G$ be a dual function satisfying the Assumption \ref{a.limit}
and let $\widehat{G}$ be as in {\rm{(}}\ref{2.13+}{\rm{)}}.
Then, for every $c>0$,
$$
\lim_{n\rightarrow\infty}V^c_n=
\sup_{\alpha\in {\cal{A}}^c} J(S_\alpha),
$$
\begin{equation}
\label{2.15}
J(S_\alpha):=
\mathbb{E}^W
\left[F(S_{\alpha})-\int_{0}^1
\widehat{G}\left(t, S_\alpha, a(t:S_\alpha)S_\alpha(t)\right) dt\right],
\end{equation}
where $\mathbb{E}^W$ denotes the expectation with respect to
$\mathbb{P}^W$.
\end{thm}

Since
$V_n^c \le V_n$ for every $c>0$,
we have the following immediate corollary.
\begin{cor}\label{cor2.1}
\begin{equation}
\label{2.115}
\lim\inf_{n\rightarrow\infty}V_n \geq
\sup_{\alpha\in {\cal{A}}}\mathbb{E}^W
\left[F(S_{\alpha})-\int_{0}^1
\widehat{G}\left(t, S_\alpha, a(t:S_\alpha)
S_\alpha(t)\right) dt\right],
\end{equation}
where ${\cal{A}}$ is the set of all bounded,
progressively measurable processes.
\end{cor}

A natural question which for now remains open
is under which assumptions the above inequality
is in fact an equality.  For the specific quadratic
penalty and  Markovian pay-offs, \cite{GS}
proves the equality.

\begin{rem}[Liquidity Premium]
\label{rem.liquidity}
{\rm{
It is an interesting question whether
the limiting super-replication cost contain
 liquidity premium. Namely, whether
the right hand side of (\ref{2.115}) is strictly bigger than $V_{BS}(F)$.
For Markovian non-affine pay-offs
it was proved in \cite{CST}.
Notice that,
the standard Black--Scholes price is given by
$V_{BS}(F):=\mathbb{E}^W F(S_{\sigma})$
and this can be achieved by simply setting
the control $\alpha\equiv\sigma$
in the right hand side of (\ref{2.115}).

In the generality considered in this paper,
the following argument might be utilized to establish
liquidity premium. Fix $\epsilon>0$.
From (\ref{2.13}), one can prove the following estimate
$$
\sup_{\alpha\in {\cal{A}}^\epsilon} \mathbb{E}^W
\left[\int_{0}^1
G\left(t, S_\alpha,a(t:S_\alpha)
S_\alpha(t)\right) dt\right]= O(\epsilon^2).
$$
Thus in order to prove the strict inequality, it remains
 to show that there exists a constant $C>0$ such that
$$
\sup_{\alpha\in {\cal{A}}^\epsilon}
 \mathbb{E}^W \left[ F(S_{\alpha})\right] \geq
 \mathbb{E}^W F(S_{\sigma})+ C\epsilon.
$$
Notice that $\sup_{\alpha\in {\cal{A}}^\epsilon} \mathbb{E}^W F(S_{\alpha})$
is exactly the $G$-expectation
of Peng. For many classes of pay-offs,
this methodology can be used to prove the existence
of a liquidity premium.
Indeed for convex type of pay-offs
such as put options, call options, Asian (put or call) options
this can be verified directly, by observing that
the maximum in
the above expression
is achieved for
$\alpha\equiv\sqrt{\sigma(\sigma+2\epsilon)}$.
\qed}}\end{rem}

We close this section by revisiting the Example \ref{e.powers}.

\begin{example}
{\rm Let $g_\gamma$ be the power penalty function given
in Example \ref{e.powers}.  In the case of $\gamma =2$,
$\widehat{G}$ is also a quadratic function.  Hence the limit
stochastic optimal control problem
is exactly the one derived and studied in \cite{CST,GS}.
The case $\gamma >2$ is not covered by our hypothesis
but formally the limit value function is equal to the Black-Scholes
price as $\widehat{G}$ is finite and zero only when $\alpha \equiv \sigma$.
This result can be proved from
our results by appropriate approximation arguments.
The case $\gamma \in [1,2)$
is included in our hypothesis and
the limit of the truncated problem is the $G$-expectation.  Namely,
only volatility processes $\alpha$ that are in a certain interval
are admissible.

Since in these markets
the investors make only small transactions,
larger $\gamma$ means
less trading cost.  Hence when $\gamma$
is sufficiently large (i.e., $\gamma >2$),
then the trading penalty is completely
avoided in the limit.  Hence for these values of $\gamma$,
the limiting super-replication cost is simply
the usual replication price in a complete market.}
\qed
\end{example}

\section{Duality}
\label{sec4}\setcounter{equation}{0}
In this section, we prove the duality result
Theorem \ref{thm2.1}. Fix $n\in\mathbb{N}$ and consider the
$n$-step binomial model with the penalty function $g$.
We first motivate the result and prove
one of the inequalities.  Then, the proof is completed
by casting the super-replication problem
as a convex program and using the standard duality.
Indeed, for any $k=0,\ldots, n-1$,
$$
Y^\pi(k+1)= Y^\pi(k) + \gamma(k+1) [S^{(n)}(k+1)-S^{(n)}(k)]
-g\left(\frac{k}{n}, \gamma(k+1)-\gamma(k)\right).
$$
Since $\gamma(0)=0$ and $Y^\pi(0)=x$,
we sum over $k$ to arrive at
\begin{eqnarray*}
Y^\pi(n) &=& x + \sum_{k=0}^{n-1}
\left( \gamma(k+1) [S^{(n)}(k+1)-S^{(n)}(k)]
-g\left(\frac{k}{n}, \gamma(k+1)-\gamma(k)\right)\right)\\
&=& x + \sum_{k=0}^{n-1}
\left( [\gamma(k+1)-\gamma(k)]\ [S^{(n)}(n)-S^{(n)}(k)]
-g\left(\frac{k}{n}, \gamma(k+1)-\gamma(k)\right)\right).
\end{eqnarray*}
Let $\mathbb{P}$ be a probability measure
in $\mathcal{Q}_n$.  We take the conditional expectations
and use the definition of the dual function $G$
to obtain,
\begin{eqnarray*}
\mathbb{E}^{\mathbb{P}}[Y^\pi(n)]
&=& x + \mathbb{E}^{\mathbb{P}}\left(  \sum_{k=0}^{n-1}
 [\gamma(k+1)-\gamma(k)]\ [\mathbb{E}^{\mathbb{P}}(S^{(n)}(n)| {\cal{F}}_k)-S^{(n)}(k)]
-g\left(\frac{k}{n}, \gamma(k+1)-\gamma(k)\right)\right)\\
&\le &
x +\mathbb{E}^{\mathbb{P}}\left( \sum_{k=0}^{n-1}
 G\left(\frac{k}{n}, \mathbb{E}^{\mathbb{P}}(S^{(n)}(n)| {\cal{F}}_k)-S^{(n)}(k)
\right)\right).
\end{eqnarray*}
If $\pi$ is a super-replicating strategy
with initial wealth $x$,
then $Y^\pi(n) \ge F_n$ and
\begin{eqnarray*}
x &\ge &
\mathbb{E}^{\mathbb{P}}
\left( F_n - \sum_{k=0}^{n-1}
 G\left(\frac{k}{n}, \mathbb{E}^{\mathbb{P}}(S^{(n)}(n)| {\cal{F}}_k)-S^{(n)}(k)
\right)\right).
\end{eqnarray*}
Since $\mathbb{P}\in \mathcal{Q}_n$ is arbitrary,
the above calculation proves that
$$
V_n \ge \sup_{\mathbb{P}\in \mathcal{Q}_n}
\mathbb{E}^{\mathbb{P}}
\left( F_n - \sum_{k=0}^{n-1}
G\left(\frac{k}{n}, \mathbb{E}^{\mathbb{P}}(S^{(n)}(n)| {\cal{F}}_k)-S^{(n)}(k)
\right)\right).
$$
The opposite inequality is proved using
the standard duality.  Indeed, the
proof that follows do not use the above
calculations.
\vspace{10pt}

\noindent
{\em Proof of Theorem \ref{thm2.1}.}

We model the $n$-step binomial model as in \cite{CJ}.
Consider a tree whose nodes
are sequences of the form $(a_1,...,a_k)\in\{-1,1\}^k$, $0\leq k\leq n$.
The set of all nodes will be denoted by $\mathbb{V}$.
The empty sequence (corresponds to the case $k=0$) is the
root of the tree and will be denoted by $\emptyset$.
In
our model each node of the form
$u=(u_1,...,u_k)\in \{-1,1\}^k$, $k<n$ has two
immediate successors $(u_1,...,u_k,1)$ and $(u_1,...,u_k,-1)$.
Let
$\mathbb{T}:=\{-1,1\}^n$ be the set of all terminal nodes.
For $u\in\mathbb{V}\setminus \mathbb{T}$,
denote by $u^{+}$ the set which consists of
the immediate successors of $u$.
The unique immediate predecessor of a node
$u=(u_1,...,u_k)\in\mathbb{V}\setminus\{\emptyset\}$
is denoted by
$u^{-}:=(u_1,...,u_{k-1})$.
For $u=(u_1,...,u_k)\in\mathbb{V}\setminus\mathbb{T}$,
let
$$
\mathbb{T}(u):=\{v\in\mathbb{T}|v_i=u_i \ \ \forall 1\leq i\leq k\},
$$
with $\mathbb{T}(\{\emptyset\})=\mathbb{T}$.
For $u\in\mathbb{V}$,
 $l(u)$ is the number of elements in the sequence $u$,
 where we set $l(\emptyset)=0$.
Finally, we define the functions $S:\mathbb{V}\rightarrow\mathbb{R}$,
$\hat{S}:\mathbb{V}\rightarrow\mathcal{C}^{+}[0,1]$ and
$\hat{F}:\mathbb{T}\rightarrow\mathbb{R}_{+}$ by
\begin{eqnarray*}
& S(u)= s_0\exp\left(\frac{\sigma}{\sqrt{n}}\sum_{i=1}^{l(u)} u_i\right), \ \
\hat{S}(u)= \mathcal{W}_n(\{S(u_1,...,u_{k\wedge\l(u)})\}_{k=0}^n)
\\& \hat{F}(v)=F(\hat{S}(v)), \ \   \forall u\in \mathbb{V},  \ \ v\in\mathbb{T}.
\end{eqnarray*}
In this notation, the super-replication cost
$V_n$ is the solution of the following
convex minimization problem
\begin{equation}\label{4.5}
{\mbox{minimize}} \ Y(\emptyset)
\end{equation}
over all $\beta, \gamma, Y$
subject to the constrains
\begin{equation}\label{4.6}
\gamma(\emptyset)=0,
\end{equation}
\begin{equation}
\label{4.7}
\gamma({u})-\gamma({u^{-}})-\beta(u^{-})=0, \ \ \forall u\in
\mathbb{V}\setminus {\{\emptyset\}},
\end{equation}
\begin{equation}
\label{4.8}
Y(u)+g\left(\frac{l(u^{-})}{n},\hat{S}(u^{-}),\beta(u^{-})\right)-
\gamma(u)[S(u)-S(u^{-})]-Y(u^{-})\leq 0, \ \
\forall u\in \mathbb{V}\setminus {\{\emptyset\}},
\end{equation}
\begin{equation}\label{4.9}
Y(u) \geq \hat{F}(u), \ \ \forall{u}\in\mathbb{T}.
\end{equation}
Notice that (\ref{2.7}) implies that
the constraint (\ref{4.8}) should be in fact an equality.
However,
this modification of the constraint does not
alter the value of the optimization problem.
The optimization problem which is given by
(\ref{4.5})--(\ref{4.9}) is an ordinary convex
program on the space $\mathbb{R}^{\mathbb{V}\setminus\mathbb{T}}
\times\mathbb{R}^{\mathbb{V}}\times\mathbb{R}^{\mathbb{V}}$.
Following the Kuhn-Tucker theory (see \cite{R}) we define the Lagrangian
$L:\mathbb{R}^\mathbb{V}\times\mathbb{R}^\mathbb{{V}\setminus {\{\emptyset\}}}_{+}
\times\mathbb{R}^{\mathbb{T}}_{+}\times\mathbb{R}^{\mathbb{V}\setminus\mathbb{T}}\times
\mathbb{R}^{\mathbb{V}}\times\mathbb{R}^{\mathbb{V}}\rightarrow\mathbb{R}$ by
\begin{eqnarray}
&\nonumber
& L(\Upsilon, \Phi,\Theta,\beta,\gamma,Y)= Y(\emptyset)+\Upsilon(\emptyset) \gamma(\emptyset)+
\sum_{u\in \mathbb{V}\setminus {\{\emptyset\}}}\Upsilon(u)
\left(\gamma({u})-\gamma({u^{-}})-\beta(u^{-})\right)\\
\nonumber
&&\quad\qquad + \sum_{u\in \mathbb{V}\setminus {\{\emptyset\}}}
\Phi(u)\left(Y(u)+g\left(\frac{l(u^{-})}{n},\hat{S}(u^{-}),\beta(u^{-})\right)
-\gamma(u)\left(S(u)-S(u^{-})\right)-Y(u^{-})\right)\\
\nonumber
&&\quad\qquad + \sum_{u\in\mathbb{T}}\Theta(u)(\hat{F}(u)-Y(u)).
\end{eqnarray}
We rearrange the above expressions to arrive at
\begin{eqnarray}
\label{4.11}
&&L(\Upsilon, \Phi, \Theta,\beta,\gamma,Y)= Y(\emptyset)(1-\sum_{u\in \emptyset^{+}}\Phi(u))
+\sum_{u\in\mathbb{V}\setminus(\{\emptyset\}
\cup \mathbb{T})}
Y(u)\big(\Phi(u)-\sum_{\tilde{u}\in u^{+}}\Phi(\tilde{u})\big)\\
\nonumber
&&\quad \qquad +\sum_{u\in\mathbb{T}}Y(u)(\Phi(u)-\Theta(u))+
\gamma(\emptyset)\big(\Upsilon(\emptyset)-\sum_{u\in \emptyset^{+}}\Upsilon(u)\big)\\
\nonumber
&&\quad \qquad +\sum_{u\in\mathbb{V}\setminus\{\emptyset\}}\gamma(u)\bigg(\Upsilon(u)-
\sum_{\tilde u\in u^{+}}\Upsilon(\tilde u)\Phi(u)\big(S(u)-S(u^{-})\big)\bigg)
+\sum_{u\in\mathbb{T}}\Theta(u)\hat{F}(u)\\
\nonumber
&&\quad \qquad +
\sum_{u\in\mathbb{V}\setminus\mathbb{T}}\left(\sum_{\tilde{u}\in\ u^{+}}\Phi(\tilde u)
g\left(\frac{l(u)}{n},\hat{S}(u),\beta(u)\right)-\beta(u)\sum_{\tilde{u}\in\ u^{+}}\Upsilon(\tilde u)\right).
\end{eqnarray}
By Theorem 28.2 in \cite{R}, we conclude that
the value of the optimization problem  (\ref{4.5})-(\ref{4.9}) is also
equal to
\begin{equation}
\label{4.12}
V_n=\sup_{(\Upsilon, \Phi, \Theta)\in \mathbb{R}^\mathbb{V}\times
\mathbb{R}^\mathbb{{V}\setminus {\{\emptyset\}}}_{+}
\times\mathbb{R}^{\mathbb{T}}_{+}}
\quad \inf_{(\beta,\gamma,Y)\in \mathbb{R}^{\mathbb{V}
\setminus\mathbb{T}\times}
\mathbb{R}^{\mathbb{V}}\times\mathbb{R}^{\mathbb{V}}}
L(\Upsilon, \Phi, \Theta,\beta,\gamma,Y).
\end{equation}
Using (\ref{4.11}) and (\ref{4.12}), we conclude that
\begin{eqnarray}\label{4.13}
V_n&=&\sup_{(\Upsilon, \Phi, \Theta)\in D}\inf_{(\beta,\gamma,Y)\in \mathbb{R}^{\mathbb{V}
\setminus\mathbb{T}\times}
\mathbb{R}^{\mathbb{V}}\times\mathbb{R}^{\mathbb{V}}}
\bigg[\sum_{u\in\mathbb{T}}\Theta(u)\hat{F}(u)\\
&&+
\left.\sum_{u\in\mathbb{V}\setminus\mathbb{T}}\left(\sum_{\tilde{u}\in\ u^{+}}\Phi(\tilde u)
g\left(\frac{l(u)}{n},\hat{S}(u),\beta(u)\right)-
\beta(u)\sum_{\tilde{u}\in\ u^{+}}\Upsilon(\tilde u)\right)\right]\nonumber
\end{eqnarray}
where $D\subset \mathbb{R}^\mathbb{V}\times\mathbb{R}^\mathbb{{V}\setminus {\{\emptyset\}}}
\times\mathbb{R}^{\mathbb{T}}_{+}$
is the subset of all  $(\Upsilon,\Phi,\Theta)$ satisfying the constraints
\begin{equation}\label{4.14}
\begin{split}
\sum_{u\in \emptyset^{+}}\Phi(u)=1, \ \ \sum_{\tilde u\in\ u^{+}}\Phi(\tilde u)=\Phi(u),
\ \ \forall u\in\mathbb{V}\setminus(\mathbb{T}\cup{\{\emptyset\}}),
\end{split}
\end{equation}
\begin{equation}\label{4.14+}
\begin{split}
\Upsilon(u)=\Phi(u)(S(u)-S(u^{-}))+\sum_{\tilde u\in\ u^{+}}\Upsilon(\tilde u),
\ \ \forall{u}\in\mathbb{V}\setminus\{\emptyset\},
\end{split}
\end{equation}
\begin{equation}\label{4.14++}
\Phi(u)=\Theta(u), \ \ \forall{u}\in \mathbb{T}.
\end{equation}
By (\ref{4.14})-(\ref{4.14+}), we obtain that for any $(\Upsilon, \Phi, \Psi)\in D$,
\begin{eqnarray}\label{4.15}
\frac{\sum_{\tilde u\in\ u^{+}}\Upsilon(\tilde u)}
{\sum_{\tilde u\in\ u^{+}}\Phi(\tilde u)}=
\frac{\sum_{\tilde u\in\mathbb{T}(u)}\Phi(u)S(u)}{\Phi(u)}-S(u),
 \ \ \forall u\in \mathbb{V}\setminus\mathbb{T},
\end{eqnarray}
where we use the convention that $0/0=0$
(observe that if $\Phi(u)=0$ then $\sum_{\tilde u\in\mathbb{T}(u)}\Phi(\tilde u)S(\tilde u)=0$).
Let $\mathbb{D}\subset \mathbb{R}^{\mathbb{V}\setminus\{\emptyset\}}_{+}$
be the set of all functions
$\Phi:\mathbb{V}\setminus\{\emptyset\}\rightarrow\mathbb{R}_{+}$
which satisfy (\ref{4.14}).
In view of (\ref{2.6}), (\ref{4.13})-(\ref{4.14}) and (\ref{4.14++})-(\ref{4.15}),
\begin{equation}\label{4.16}
V_n=\sup_{\Phi\in \mathbb{D}}\sum_{u\in\mathbb{T}}\Phi(u)\bigg(\hat{F}(u)-
G\bigg(\frac{l(u)}{n},\hat{S}(u),\frac{\sum_{\tilde u\in\mathbb{T}(u)}
\Phi(\tilde u)S(\tilde u)}{\Phi(u)}-S(u)\bigg)\bigg).
\end{equation}
Clearly there is a natural bijection $\pi:\mathbb{D}\rightarrow\mathcal{Q}_n$
(where, recall $\mathcal{Q}_n$ is the set of
all probability measures on $(\Omega,\mathcal{F}_n)$)
such that for any $\Phi\in\mathbb{D}$ the probability
measure $P:=\pi(\Phi)$ is given by
\begin{equation}\label{4.16+}
\mathbb{P}(\xi_1=u_1,\xi_2=u_2,...,\xi_n=u_n)=\Phi(u),
\ \ \forall u=(u_1,...,u_n)\in\mathbb{T}.
\end{equation}
Finally we combine (\ref{4.16}) and (\ref{4.16+}) to  conclude that
$$
V_n=\sup_{\mathbb{P}\in \mathcal{Q}_n}
\mathbb{E}^\mathbb{P}\bigg(F_n-
\sum_{k=0}^{n-1}G\left(\frac{k}{n},\mathcal{W}_n(S^{(n)}(k)),
\mathbb{E}^\mathbb{P}(S^{(n)}(n)|\mathcal{F}_k)-S^{(n)}(k)\right)\bigg).
$$
\qed

\section{Proof of Theorem \ref{thm2.2}}
\label{sec5}
\setcounter{equation}{0}
In this section we prove Theorem
\ref{thm2.2}.  However,the proofs of several technical results
needed in this proof are relegated to Section \ref{sec3}.
Also the Kusuoka's construction of
discrete martingales are outlined in the next section.

We start with some definitions.
Let $B$ be the canonical map
on the space $\mathcal{C}[0,1]$, i.e., for each $t \in [0,1]$
$B(t):\mathcal{C}[0,1]\rightarrow\mathbb{R}$ is
given by $B(t)(x)=x(t)$. Next, let
$M$ be a strictly positive, continuous martingale
defined on some probability space
$(\tilde{\Omega},\tilde{F},\tilde{P})$ and satisfies
\begin{equation}
\label{3.17+}
M(0)=s_0 \ \ \mbox{and} \ \ \frac{d\langle\ln M\rangle(t)}{dt}\leq C,
\ \ \mathcal{L}\otimes\tilde{P} \ \mbox{a.s.}
\end{equation}
for some constant $C$.
For a martingale $M$ satisfying (\ref{3.17+}), we define several
related quantities.  Let
$\widehat{G}$ be as in Assumption \ref{a.limit} and
$\sigma$
be the constant volatility in the definition
of the discrete market, c.f., (\ref{2.4}).  Set
\begin{eqnarray}
\label{e.A}
A(t:M)&:=&\frac{ \langle\ln M\rangle(t) - \sigma^2 t}{2\sigma},
\qquad
a(t:M):= \frac{d}{dt} A(t:M),\\
\label{3.17++}
J(M)& =& \tilde{E}\left[F(M)-
\int_{0}^1 \widehat{G}\left(t, M,  a(t:M) M(t)\right)
 dt\right],
\end{eqnarray}
where $\tilde{E}$ is the expectation with respect to $\tilde{P}$.
Notice that the notation $a$ is consistent with
the already introduced function $a(t:S_\alpha)$ in (\ref{e.a})
and
$J(M)$ agrees  with the function defined in  (\ref{2.15}).
Also, from (\ref{2.2}), (\ref{2.13}) and (\ref{3.17+})
it follows that the right hand side
of (\ref{3.17++}) is well defined.
\vspace{10pt}

\noindent
{\em{ Upper Bound.}}

For fix $c>0$, we start by proving the upper bound of Theorem \ref{thm2.2}:
\begin{equation}\label{5.1}
\lim\sup_{n\rightarrow\infty}V^c_n\leq\sup_{\alpha\in \mathcal{A}^C}
J(S_{\alpha}).
\end{equation}
In what follows, to simplify the notation,
we assume that indices have
been renamed so that the whole sequence converges.
Let $n\in\mathbb{N}$.

By Theorem \ref{thm2.1},
 we construct  probability measures $P_n$
on $(\Omega,\mathcal{F}_n)$ such that
\begin{eqnarray}
\nonumber
V^c_n &\le &\frac{1}{n}+E_n\bigg[F\big(\mathcal{W}_n(S^{(n)})\big)-
\sum_{k=0}^{n-1}G_n^c\left(\frac{k}{n},\mathcal{W}_n(S^{(n)}),
E_n\left[S^{(n)}(n)|\mathcal{F}_k\right]-S^{(n)}(k)\right)\bigg],\\
\label{5.4}
 &= &\frac{1}{n}+E_n\bigg[F\big(\mathcal{W}_n(S^{(n)})\big)-
\sum_{k=0}^{n-1}G\left(\frac{k}{n},\mathcal{W}_n(S^{(n)}),
E_n\left[S^{(n)}(n)|\mathcal{F}_k\right]-S^{(n)}(k)\right)\bigg]
\end{eqnarray}
where $E_n$ denotes the expectation with respect to $P_n$.
In the last identity we used the form of the dual function
$G_n^c$.
Indeed, (\ref{e.Gnc}) states that
either $G_n^c=G$ or $G_n^c=+\infty$.
This argument also shows that for any $0\leq k<n$,
\begin{equation}\label{5.3}
\begin{split}
\left|E_n\left[S^{(n)}(n)|\mathcal{F}_k\right]-S^{(n)}(k)\right|\leq
\frac{c }{\sqrt n}\ S^{(n)}(k),\ \  P_n \ \ \mbox{a.s.}
\end{split}
\end{equation}
Indeed, if above does not hold, then
in view of (\ref{e.Gnc}), we would conclude that
the right hand side of (\ref{5.4}) would be equal to
negative infinity.  But it is easy to show that
 $V_n^c$ is non-negative.

For $ 0 \le k \le n$, set
\begin{eqnarray*}
M^{(n)}(k)& := & E_n\big(S^{(n)}(n)|\mathcal{F}_k\big), \\
\alpha_n(k) &:= & \frac{\sqrt{n}\xi_k(M^{(n)}(k)-S^{(n)}(k))}{S^{(n)}(k)}\\
A_n(t) & := & \int_{0}^t\alpha_n([nu])du = \frac1n \sum_{k=0}^{[nt]-1} \alpha_n(k) +
\frac{nt-[nt]}{n}\ \alpha_n([nt]).
\end{eqnarray*}
Let $Q_n$ be the joint distribution of the stochastic processes
($\mathcal{W}_n(S^{(n)}),A_n)$ under the measure $P_n$.
In view of (\ref{5.3}), the hypothesis of  Lemma
\ref{convergence} is satisfied.  Hence, there exists a subsequence (denoted by $n$ again)
and a probability measure $P$ on the probability space $\mathcal{C}[0,1]$
such that
$$
Q_n\Rightarrow Q \ \mbox{on the space} \ \mathcal{C}[0,1]\times \mathcal{C}[0,1]
$$
where  $Q$ is the joint distribution under $P$
of the canonical process $B$ and the process
$A(\cdot:B)$ defined in (\ref{e.A}).
From the Skorohod representation theorem (Theorem 3 of \cite{D})
it follows that
there exists a probability space $(\tilde{\Omega},\tilde{F},\tilde{P})$ on which
\begin{equation}
\label{5.6++}
\left(\mathcal{W}_n(S^{(n)}),A_n(\cdot)\right)\rightarrow
\left(M,A(\cdot:M)\right) \ \ \tilde{P}\mbox{-a.s.}
\end{equation}
on the space $\mathcal{C}[0,1]\times \mathcal{C}[0,1]$, where $M$ is a strictly positive martingale.
Furthermore,  (\ref{5.3}) implies that Lemma \ref{tightness} applies to this sequence.
Hence we have the following pointwise estimate,
$$
|a(t:M)|=|A^\prime(t:M)|\le c \ \ \mathcal{L}\otimes\tilde{P}\mbox{-a.s.}
$$

Next, we will replace the sequence $\alpha_n$
(which converges only weakly) by a pointwise convergent sequence.
Indeed, by Lemma A1.1 in \cite{DS}, we construct a sequence
$$
\eta_n \in conv( \tilde\alpha_n,\tilde \alpha_{n+1},...),\qquad
{\mbox{where}}\qquad \tilde\alpha_n(t):= \alpha_n([nt])
$$
such that $\eta_n$ converges almost surely in $\mathcal{L}\otimes \tilde{P}$ to a stochastic process $\eta$.
We now use (\ref{5.6++}) together with
the Lebesgue dominated convergence theorem.
The result is
\begin{eqnarray*}
\int_{0}^t \eta(u)du&=&\lim_{n\rightarrow\infty}\int_{0}^t \eta_n(u)du=
\lim_{n\rightarrow\infty}\int_{0}^t \alpha_n([nu])du\\
&=&A(t:M)= \int_{0}^t  a(u:M) du,
\qquad \qquad \mathcal{L}\otimes\tilde{P} \ \ \mbox{a.s.}
\end{eqnarray*}
Hence, we conclude that
$$
\eta(t)=a(t:M),\qquad \qquad \mathcal{L}\otimes\tilde{P} \ \ \mbox{a.s.}
$$
We are now ready to use the assumption (\ref{2.13+}).  Indeed, by definition
$$
M^{(n)}(k)-S^{(n)}(k)= \alpha_n(k) \frac{ \xi_k}{\sqrt{n}}\ S^{(n)}(k)
= \alpha_n(k) \frac{ \xi_k}{\sqrt{n}}\ \mathcal{W}_n(S^{(n)})(k/n).
$$
Also, by (\ref{5.6++}),
$\mathcal{W}_n(S^{(n)})$ converges to $M$.  Hence in view of
(\ref{5.3}), we can use (\ref{2.13+}) to conclude that
$$
\lim_{n\rightarrow\infty} \left | n G\left(\frac{[nt]}{n},\mathcal{W}_n(S^{(n)}),
M^{(n)}([nt])-S^{(n)}([nt])\right)-
\widehat{G} (t,M,\alpha_n([nt])M(t)) \right |=0, \ \ \mathcal{L}\otimes\tilde{P}
 \ \mbox{a.s.}
 $$
The estimate (\ref{3.3}) and the growth assumption (\ref{2.13}) imply
that the above sequences are uniformly integrable.
Therefore,
\begin{eqnarray*}
I&=& \lim_{n\rightarrow\infty}E_n\left[\sum_{k=0}^{n-1}G\left(
\frac{k}{n},\mathcal{W}_n(S^{(n)}),E_n\left[S^{(n)}(n)|\mathcal{F}_k
\right]-S^{(n)}(k)\right)\right]\\
&=& \lim_{n\rightarrow\infty}E_n\left[ \int_{0}^1
n G\left(\frac{[nt]}{n},\mathcal{W}_n(S^{(n)}),M^{(n)}([nt])-S^{(n)}([nt])\right)dt \right]\\
&=& \lim_{n\rightarrow\infty}E_n\left[ \int_{0}^1
\widehat{G} \left(t,M, \alpha_n([nt])M(t)\right) dt\right],
\end{eqnarray*}
where again, without loss of generality (by passing to a subsequence)
we assumed that the above limits exist.
We now use the convexity of $\widehat{G}$ with respect to third variable
(in fact, $\widehat{G}$ is quadratic in $y$) together with the uniform integrability
(which again follows from (\ref{2.13}) and Lemma \ref{tightness}) and the Fubini theorem.
The result is
\begin{eqnarray*}
I&=& \lim_{n\rightarrow\infty}E_n\left[ \int_{0}^1
\widehat{G} \left(t,M, \alpha_n([nt])M(t)\right) dt\right]
= \lim_{n\rightarrow\infty}\tilde{E}\left[ \int_{0}^1
\widehat{G} \left(t,M, \alpha_n([nt])M(t)\right) dt\right]\\
&\ge & \lim_{n\rightarrow\infty}\tilde{E}\left[ \int_{0}^1
\widehat{G} \left(t,M, \eta_n(t)M(t)\right) dt\right]\\
&=& \tilde{E}\left[ \int_{0}^1
\widehat{G} \left(t,M, \eta(t)M(t)\right) dt\right]
= \tilde{E}\left[ \int_{0}^1
\widehat{G} \left(t,M, a(t:M) M(t)\right) dt\right].
\end{eqnarray*}

The growth assumption on $F$, namely  (\ref{2.2}) and Lemma \ref{tightness},
also imply that
the sequence $F\big(\mathcal{W}_n(S^{(n)})\big)$
is uniformly integrable.  Then, by (\ref{5.6++}),
$$
\lim_{n\rightarrow\infty}E_nF\big(\mathcal{W}_n(S^{(n)})\big)=\tilde{E}F(M).
$$
Hence, we have shown that
\begin{eqnarray*}
\limsup_{n \to \infty} V^c_n &\le & \limsup_{n \to \infty}
E_n\bigg[F\big(\mathcal{W}_n(S^{(n)})\big)-
\sum_{k=0}^{n-1}G\left(\frac{k}{n},\mathcal{W}_n(S^{(n)}),
E_n\left[S^{(n)}(n)|\mathcal{F}_k\right]-S^{(n)}(k)\right)\bigg]\\
&\le &
\tilde{E} \left[ F(M)- \int_{0}^1
\widehat{G} \left(t,M,a(t:M) M(t)\right) dt\right]
= J(M).
\end{eqnarray*}
The above  together with
Lemma \ref{randomization} yields (\ref{5.1}).
\vspace{15pt}

\noindent
{\em Lower Bound.}

Let $\mathcal L({c})$ be the class of all
adapted volatility processes given in Definition \ref{d.vol}.
In Lemma \ref{density} below, it is shown that
this class is dense.  Hence for the lower bound it
is sufficient to prove that for any $\alpha\in\mathcal L({c})$,
\begin{equation}
\label{e.vnc}
\lim_{n\rightarrow\infty}V^c_n\geq J(S_{\alpha}).
\end{equation}

Our main tool is the Kusuoka construction which is summarized in
Theorem \ref{Kusuoka}.

We fix  $\alpha\in\mathcal L({c})$.
Let $P^{(\alpha)}_n$,
$\kappa_n^{(\alpha)}$ and $M_n^{(\alpha)}$
be as in Theorem \ref{Kusuoka}.
In view of the definition of $M^{(\alpha)}_n$, (\ref{3.39-}),
and the bounds on $\kappa^{(\alpha)}_n$,
the following estimate holds for all sufficiently large $n$,
$$
|{M}^{(\alpha)}_n(k)-S^{(n)}(k)|\leq
\frac{c }{\sqrt n}\ S^{(n)}(k), \ \ \forall{k}, \ \ P^{(\alpha)}_n \ \ \mbox{a.s.}
$$
By the dual representation and the above estimate,
\begin{equation}
\label{e.vnc1}
\lim_{n\rightarrow\infty}V^c_n\geq
\lim\sup_{n\rightarrow\infty}
E^{(\alpha)}_n\left[F\left(\mathcal{W}_n(S^{(n)})\right)-
\sum_{k=0}^{n-1}G\left(\frac{k}{n},\mathcal{W}_n(S^{(n)}),M^{(\alpha)}_n(k)-S^{(n)}(k)\right)\right]
\end{equation}
where $E^{(\alpha)}_n$ denotes the expectation with respect to $P^{(\alpha)}_n$.
From Theorem \ref{Kusuoka} and the Skorohod representation theorem it follows that
there exists a probability space $(\tilde\Omega,\tilde{F},\tilde{P})$
on which
\begin{equation}
\label{5.26}
\left( \mathcal{W}_n(S^{(n)}), \mathcal{W}_n(\kappa^{(\alpha)}_n)\right)
\rightarrow \left(S_{\alpha}, a(\cdot:S_\alpha)\right)\ \ \tilde{P} \ \ \mbox{a.s.}
\end{equation}
on the space $\mathcal{C}[0,1]\times\mathcal{C}[0,1]$.
Recall that the quadratic variation density $a$ is defined in (\ref{e.a})
and also in (\ref{e.A}).
We argue exactly as in the upper bound to show that
$$
\lim_{n\rightarrow\infty}{E}^{(\alpha)}_n
F\big( \mathcal{W}_n(S^{(n)})\big)=\tilde{E}F(S_{\alpha}).
$$

Finally, we need to connect the difference $(M^{(\alpha)}_n-S^{(n)})$ to
$\kappa^{(\alpha)}_n$ and therefore to $a(\cdot:S_\alpha)$ through
(\ref{5.26}).  Indeed, in view of the definition (\ref{3.39-}),
\begin{eqnarray*}
\sqrt{n}\xi_{k}(M^{(\alpha)}_n(k)-S^{(n)}(k))
&= &\sqrt{n}\xi_{k} S^{(n)}(k) \left(
\exp\left(\xi_k\kappa^{(\alpha)}_n(k)n^{-1/2}\right) -1 \right)\\
&=&  S^{(n)}(k)  \kappa^{(\alpha)}_n(k) + o(n^{-1/2}).
\end{eqnarray*}
In the approximation above, we used the fact that
$\kappa^{(\alpha)}$'s are uniformly bounded
by construction.
We now use (\ref{5.26}) to arrive at
\begin{equation}
\label{5.26++}
\lim_{n\rightarrow\infty} \sqrt{n}\xi_{[nt]}(M^{(\alpha)}_n([nt])-S^{(n)}([nt]))=a(t:S_\alpha) S_{\alpha}(t),
\ \ \mathcal{L}\otimes\tilde{P} \ \ \mbox{a.s.}
\end{equation}
As in the upper bound case,
the growth condition (\ref{2.13}) and Lemma \ref{tightness} imply that
the sequences
$$
n G({[nt]}/n,\mathcal{W}_n(S^{(n)}),M^{(\alpha)}_n([nt])-S^{(n)}([nt]))\quad
{\mbox{and}}\quad
\widehat{G}\left(t,S_{\alpha},a(t:S_\alpha)S_{\alpha}(t)\right),
$$
are uniformly integrable in $\mathcal{L}\otimes\tilde{P}$.
Since $\widehat{G}$ is continuous by the
Fubini's theorem and (\ref{2.13+}), (\ref{5.26}), (\ref{5.26++}),
we obtain,
\begin{eqnarray*}
\tilde I&:=&
\lim_{n\rightarrow\infty}E_n\left[\sum_{k=0}^{n-1}G\left(
\frac{k}{n},\mathcal{W}_n(S^{(n)}),M^{(\alpha)}_n(k)-S^{(n)}(k)\right)\right]\\
&=& \lim_{n\rightarrow\infty}\tilde E \left[\int_{[0,1]}
n G\left(\frac{[nt]}{n},\mathcal{W}_n(S^{(n)}),M^{(\alpha)}_n([nt])-S^{(n)}([nt])\right)dt
\right]\\
&=& \lim_{n\rightarrow\infty}\tilde E \left[\int_{[0,1]}
\widehat{G}\left(t,S_{\alpha},a(t:S_\alpha)S_{\alpha}(t)\right)dt \right].
\end{eqnarray*}
We use the above limit results for $\tilde I$ and for $F_n$ in (\ref{e.vnc1}).
The resulting inequality is exactly (\ref{e.vnc}).
Hence the proof of the lower bound is also complete.
\qed

\section{Kusuoka's construction}\label{secKusuoka}\setcounter{equation}{0}

In this section, we fix a martingale
$S_\alpha$ given  by (\ref{2.14}).
Then, the main goal of this section  is to
construct a sequence of martingales
on the discrete space that
approximate $S_\alpha$.
We also require the quadratic variation of $S_\alpha$ to be
approximated as well.

In \cite{K} Kusuoka
provides an elegant approximation
for sufficiently smooth volatility process $\alpha$.
Here we will only state the results of Kusuoka
and refer to \cite{K} for the construction.
We start by defining the class
of ``smooth" volatility processes. As before, let
$(\Omega_W,\mathcal{F}^{W},\mathbb{P}^{W}$) be a  Brownian probability space
and $W$ be the standard Brownian motion.

\begin{dfn}
\label{d.vol}
{\rm
For a fixed constant $c>0$,
 $\mathcal{L}(c)\subset\mathcal{A}^c$ is the set of all
adapted processes $\alpha$ on the
Brownian space $(\Omega_W,\mathcal{F}^{W},\mathbb{P}^{W}$)
which are given by
$$
\alpha(t):=\alpha(t,\omega)=f(t,W(\omega)),
\qquad
(t,\omega) \in [0,1] \times \Omega_W,
$$
where $f:[0,1]\times \mathcal{C}[0,1]\rightarrow\mathbb{R}_{+}$
is a bounded function which satisfies the following conditions.
\vspace{5pt}

\noindent
i). For any $t \in [0,1]$, if two $S, \tilde S \in  \mathcal{C}[0,1]$
satisfy $S(u)=\tilde S(u)$ for all $u \in [0,t]$, then
$f(t,S)=f(t,\tilde S)$.
(This simply means that $\alpha$ is adapted.)
\vspace{5pt}

\noindent
ii). There is $\delta(f)>0$ such that for all $(t,S) \in [0,1] \times
 \mathcal{C}[0,1]$,
$$
\left|\frac{f^2(t,S) - \sigma^2}{2 \sigma}\right| \le
c- \delta(f),
$$
and
\begin{equation}
\label{3.37+}
 f(t,S)=\sigma, \ \ \mbox{if} \ \ t>1-\delta(f).
\end{equation}
\vspace{2pt}

\noindent iii).  There is $L(f)>0$ such that for all
$(t_1,t_2)\in [0,1]$, $S, \tilde S \in  \mathcal{C}[0,1]$,
$$
\left|f(t_1,S)-f(t_2,\tilde S)\right|\leq L(f)
\left(|t_1-t_2|+\|S-\tilde S\|_\infty\right).
$$
}
\qed
\end{dfn}

In Kusuoka's construction
the condition (\ref{3.37+}) is not needed.
However, this regularity
allows us to control
the behavior of the martingales near maturity .

Recall from Section \ref{sec:2} that $\Omega=\{1,1\}^\infty$,
$\xi$ is the canonical map (i.e., $\xi_k(\omega)=\omega_k$)
and $\mathbb{Q}$ is the symmetric product measure.
The martingales constructed in \cite{K} are of the form
\begin{equation}
\label{3.39-}
M^{(\alpha)}_n(k,\omega):=S^{(n)}(k,\omega)
\exp\left(\xi_k(\omega)\kappa^{(\alpha)}_n(k,\omega)n^{-1/2}\right), \ \ 0\leq k\leq n,\
\omega \in \Omega,
\end{equation}
where the sequence of discrete {\em predictable}
processes $\kappa^{(\alpha)}_n$ need to be constructed.
Now let $P^{(\alpha)}_n$ be a measure on $\Omega$
such that
the process $M_n^{(\alpha)}$ is a  $P^{(\alpha)}_n$-martingale.
Since, $\kappa_n^\alpha$
will be constructed as predictable processes,
a direct calculation shows that
on the $\sigma$-algebra $\mathcal{F}_n$,
this martingale measure is given by,
$$
\frac{dP^{(\alpha)}_n}{d\mathbb{Q}}(\omega)
=2^n \prod_{k=1}^n \tilde{q}^{(\alpha)}_n(k,\omega),
$$
where for $0 \le k \le n$, $\omega \in \Omega$,
\begin{eqnarray*}
\tilde{q}^{(\alpha)}_n(k,\omega)&=&
q^{(\alpha)}_n(k,\omega)\mathbb{I}_{\{\xi_k(\omega)=1\}}
+(1-q^{(\alpha)}_n(k,\omega))\mathbb{I}_{\{\xi_k(\omega)=-1\}}, \\
q^{(\alpha)}_n(k,\omega)&=&
\frac{\exp\left(\xi_{k-1}\kappa^{(\alpha)}_n(k-1,\omega) n^{-1/2}\right)
 -\left( \exp\left(\sigma n^{-1/2}\right) e^{(\alpha)}_n(k,\omega)\right)^{-1}}{
\exp\left(\sigma n^{-1/2}\right)
e^{(\alpha)}_n(k,\omega) -\left( \exp\left(\sigma n^{-1/2}\right)e^{(\alpha)}_n(k,\omega)\right)^{-1}}\\
e^{(\alpha)}_n(k,\omega)&= &\exp\left(\kappa^{(\alpha)}_n(k,\omega)n^{-1/2}\right).
\end{eqnarray*}
We require that $\kappa^{(\alpha)}_n$
is constructed to satisfy,
\begin{eqnarray*}
|\kappa^{(\alpha)}_n(k,\omega)|&<&c-\delta, \ \ \kappa^{(\alpha)}_n(k,\omega)>\delta-\frac{1}{2},\\
|\kappa^{(\alpha)}_n(k-1,\omega)-\kappa^{(\alpha)}_n(k,\omega)|&\leq& \frac{L}{\sqrt n}, \ \ 1\leq k\leq n,
\end{eqnarray*}
with constants $L, \delta>0$ independent of $n$ and $\omega$.
This regularity conditions
on $\kappa^{(\alpha)}_n$
imply that for all sufficiently large $n$, $q_n(k,\omega)\in(0,1)$
for all $k \le n$ and $\omega \in \Omega=\{1,1\}^\infty$ .
Hence, $P^{(\alpha)}_n$ is indeed a probability measure.

We also require
$$
\kappa^{(\alpha)}_n(n,\omega)=0  \ \ \mbox{for} \  \mbox{sufficiently} \  \mbox{large} \  n,
$$
to ensure $M^{(\alpha)}_n(n)=S^{(n)}(n)$.

Let $Q^{(\alpha)}_n$
be the joint distribution  of the pair
$(\mathcal{W}_n(S^{(n)}),\mathcal{W}_n(\kappa^{(\alpha)}_n)$
under $P^{(\alpha)}_n$
on the space $\mathcal{C}[0,1]\times \mathcal{C}[0,1]$
with the uniform topology.

Recall once again that the probability space
is $\Omega=\{-1,1\}^\infty$ and the filtration $\{\mathcal{F}_k\}_{k=0}^n$
is the usual one generated by the canonical map
and  that the quadratic variation density process $a(\cdot:S_\alpha)$
is given in (\ref{e.a}) as
$$
a(t:S_\alpha) = \frac{\alpha^2(t)-\sigma^2}{2\sigma}.
$$

\begin{thm}[Kusuoka \cite{K}]\label{Kusuoka}
Let $c>0$ and $\alpha\in\mathcal{L}(c)$.
Then, on $\left(\Omega,\{\mathcal{F}_k\}_{k=0}^n\right)$
there exists a sequence of predictable processes
$\kappa^{(\alpha)}_n$
satisfying the above conditions, hence
there also exist sequences of martingales
$M^{(\alpha)}_n$ and
martingale measures $P^{(\alpha)}_n$
so that
$$
Q_n^{(\alpha)}\Rightarrow Q^{(\alpha)} \ \mbox{on} \ \mbox{the} \ \
\mbox{space} \ \mathcal{C}[0,1]\times\mathcal{C}[0,1]
$$
where $Q^{(\alpha)}$ is the joint distribution of
$\left(S_{\alpha}, a(\cdot:S_\alpha)\right)$ under the Wiener measure $\mathbb{P}^W$.
\end{thm}

For the construction
of $\kappa^{(\alpha)}_n$, we refer the reader to Proposition 5.3 in \cite{K}.

\begin{rem}
\label{r.kusuoka}
{\rm It is clear that one constructs the
process $\kappa^{(\alpha)}_n$ by
an appropriate discrete approximation of $a(\cdot:S_\alpha)$.
However, this discretization is not only in time
but is also in the probability space.  Namely,
the process  $\alpha$ is a process on
the canonical probability space $\mathcal{C}[0,1]$
while  $\kappa^{(\alpha)}_n$ lives in
the discrete space $\Omega$.
This difficulty is resolved by Kusuoka in \cite{K}.}
\qed
\end{rem}

We complete this section by
stating (without proof) a lemma which
summarizes the main results from Section 4 in \cite{K};
see in particular, Lemma 4.2 and Proposition 4.27 in \cite{K}.
In our analysis the below lemma
provides the crucial tightness result which is used in the
proof of the upper bound of
Theorem \ref{thm2.2}. Furthermore, the inequality (\ref{3.3}) is
essential in establishing the uniform integrability of
several sequences.

Let $(\Omega,\mathbb{Q})$ be the probability
space introduced in Section \ref{sec:2}.

\begin{lem}
[Kusuoka \cite{K}]\label{tightness}
Let $M^{(n)}$ be a sequence positive
martingales with respect to probability measures
$P_n$ on $(\Omega,\mathcal{F}_n)$.
Suppose that there
exists a constant $c>0$ such that for any $k\leq n$,
$$
\left|S^{(n)}(k)-M^{(n)}(k)\right |\leq \frac {c S^{(n)}(k)}{\sqrt n}, \ \ P_n \ \ \mbox{a.s.}
$$
Then, for any $p>0$
\begin{equation}
\label{3.3}
\sup_{n} E_n \big(\max_{0\leq k\leq n}S^{(n)}(k)\big)^p<\infty,
\end{equation}
where ${E}_n$ is an expectation with respect to ${P}_n$.

Moreover,  the distribution $Q_n$ on $\mathcal{C}[0,1]$ of
$\mathcal{W}_n(S^{(n)})$ under $P_n$ is a tight
sequence and under any limit point
$Q$ of this sequence, the canonical process
$B$ is a strictly positive martingale in  its usual filtration.
Furthermore, the quadratic variation density
of $B$ under $Q$ satisfies,
$$
\left| a(t:B)\right| \le c, \ \ \mathcal{L}\otimes Q\mbox{-a.s.}
$$
\end{lem}

\section{Auxiliary lemmas}
\label{sec3}\setcounter{equation}{0}
In this section, we prove several results that are
used in the proof of our convergence result.
Lemmas \ref{randomization}-\ref{density} are
related to the optimal control (\ref{2.15}).
The first result, Lemma \ref{convergence}
is related to the properties of a sequence
 discrete time martingales $M^{(n)}$.
Motivated by (\ref{5.3}) and Lemma \ref{tightness},
we assume that these martingales
are sufficiently
close to the price process $S^{(n)}$.
Then, in Lemma \ref{convergence} below,
we prove that the process $\alpha_n$, defined below,
converges weakly.
The structure that we outline below is very similar to the
one constructed in Theorem \ref{Kusuoka}.  However,
below the martingales $M^{(n)}$ are
given while in the previous section they are constructed.

This limit theorem is the main tool in the proof of the upper bound
of Theorem \ref{thm2.2}.

Let  $(\Omega,\mathcal{F}_n)$
be the discrete  probability structure
given in Section \ref{sec:2}.  For
a probability measure $P_n$ $(\Omega,\mathcal{F}_n)$
and $ k\leq n$, set
\begin{eqnarray*}
M^{(n)}(k)& :=& E_n(S^{(n)}(n)|\mathcal{F}_k), \\
\alpha_n(k)&:=& \frac{\sqrt n \xi_k (M^{(n)}(k)-S^{(n)}(k))}{S^{(n)}(k)}.
\end{eqnarray*}

Suppose that there exists a constant $c>0$ such that
for any $k\leq n$,
\begin{equation}
\label{3.5}
|\alpha_n(k)|\leq c, \ \ P_n \ \ \mbox{a.s.}
\end{equation}
Let ${Q}_n$ be the distribution of $\mathcal{W}_n(S^{(n)})$
under the measure $P_n$.
Then, by Lemma \ref{tightness}
this sequence is tight.
Without loss of generality we assume that
the whole sequence $\{Q_n\}_{n=1}^\infty$ converges to
a probability measure $Q$ on $\mathcal{C}[0,1]$.
Moreover, under $Q$ the canonical map $B$ is a strictly
positive martingale.
Then,
Lemma \ref{tightness} also implies that the process
$A(\cdot:B)$ given in (\ref{e.A}) is well defined.
The next lemma proves the convergence of
the process $\alpha_n$ as well.

\begin{lem}\label{convergence}
Assume {\rm{(}}\ref{3.5}{\rm{)}}.
Let $\hat{Q}_n$ be the joint distribution
$\mathcal{W}_n(S^{(n)})$ and $\int_{0}^{t}\alpha_n([nu])du$ under $P_n$.
Then,
$$
\hat{Q}_n\Rightarrow\hat{Q} \ \mbox{on} \ \mbox{the} \
 \mbox{space} \ \mathcal{C}[0,1]\times\mathcal{C}[0,1]
$$
where $\hat{Q}$ is the
joint distribution of the canonical process $B$ and
$A(\cdot:B)$ under $Q$.
\end{lem}
\begin{proof}
Hypothesis (\ref{3.5}) imply that Lemma \ref{tightness}
apply to the sequence $P_n$.  Hence under this sequence
of measures the estimate (\ref{3.3}) holds.

Let $Y_n$ be a piecewise constant process defined by
\begin{equation}
\label{3.9}
\begin{split}
Y_n(t)=\sum_{j=1}^{[nt]} \frac{M^{(n)}(j)-M^{(n)}(j-1)}{S^{(n)}(j-1)}, \ \ t\in [0,1],
\end{split}
\end{equation}
with $Y_n(t)=0$ it $t<\frac{1}{n}$.
In view of  (\ref{3.5}), there exists a constant $c_1$ such that for any
$k<n$,
$$
\left|M^{(n)}(k+1)-M^{(n)}(k)\right |\leq \frac{c_1}{\sqrt n} S^{(n)}(k),
\qquad P_n{\mbox{-a.s.}}
$$
We use this together with (\ref{3.3}) to arrive at
\begin{equation}
\label{3.11}
\lim_{n\rightarrow\infty}E_n(\max_{1\leq k\leq n}|M^{(n)}(k)-M^{(n)}(k-1)|)=0.
\end{equation}
Let $\mathcal{D}[0,1]$ be the space of all
$c\grave{a}dl\grave{a}g$ functions equipped with the Skorohod topology (see \cite{B}).
Let $\hat{P}_n$ be the distribution
on the space $\mathcal{D}[0,1]\times\mathcal{D}[0,1]$,
of the piecewise constant process $\{(1/S^{(n)}([nt]),M^{(n)}([nt]))\}_{t=0}^1$ under the measure $P_n$.
We use (\ref{3.5}) and Lemma \ref{tightness},
to conclude that
\begin{equation}
\label{3.11+}
\hat{P}_n\Rightarrow \hat{P}  \ \mbox{on}\  \mbox{the} \ \mbox{space} \ \mathcal{D}[0,1]\times\mathcal{D}[0,1],
\end{equation}
where the measure $\hat{P}$ is the distribution of the process
$(1/B,B)$ under $Q$.
In fact, for this convergence we extend the
definition of  $B$ so that it is still  the canonical process on the space
$\mathcal{D}[0,1]$ and  the measure $Q$
is extended as a probability measure on $\mathcal{D}[0,1]$.

Since the canonical process $B$ is a
strictly positive continuous martingale under $Q$,
we apply Theorem 4.3 of \cite{DP} and use
(\ref{3.11}), (\ref{3.11+}).  The result
is the following convergence,
$$
\hat{\mathbb{Q}}_n\Rightarrow \hat{\mathbb{Q}}  \ \mbox{on}
\ \mbox{the} \mbox{space} \ \mathcal{D}[0,1]\times\mathcal{D}[0,1]\times\mathcal{D}[0,1],
$$
where $\hat{\mathbb{Q}}_n$ is the distribution of the triple
$\{(1/S^{(n)}([nt]), M^{(n)}([nt]),Y_n([nt]))\}_{t=0}^1$ under $P_n$, and $\hat{\mathbb{Q}}$
is the distribution of the triple $\left\{\left({1}/{B(t)},B(t),\int_{0}^t{dB(u)}/{B(u)}\right)\right\}_{t=0}^1$,
under the measure $Q$.

In view of the Skorohod representation theorem,
without loss of generality, we may assume that there exists a
probability space $(\tilde\Omega,\tilde{F},\tilde{P})$
and  a strictly positive continuous martingale $M$ such that
$$
\left\{\left(\frac{1}{S^{(n)}([nt])},M^{(n)}([nt]),Y_n([nt])\right)\right\}_{t=0}^1\rightarrow
\left\{\left(\frac{1}{M(t)},M(t),\int_{0}^t\frac{dM(u)}{M(u)}\right)\right\}_{t=0}^1 \ \ \tilde{P}\mbox{-a.s.}
$$
on the space $\mathcal{D}[0,1]\times \mathcal{D}[0,1]\times \mathcal{D}[0,1]$.

Now set
$Y(t)=\int_{0}^t {dM(u)}/{M(u)}$ so that
$dM=MdY$.  Therefore,
$$
M(t)=M(0)\exp\left(Y(t)-\frac{\langle Y\rangle(t)}{2}\right)\qquad
\Rightarrow \qquad
\langle \ln M\rangle(t)= \langle Y\rangle(t).
$$
Hence to complete the proof of the Lemma, it is sufficient to show that
$$
\left\{\int_{0}^t\alpha_n([nu])du\right\}_{t=0}^1\rightarrow
\left\{\frac{\langle Y\rangle(t)-\sigma^2t}{2\sigma}\right\}_{t=0}^1
 \ \
\tilde{P}\mbox{-a.s.} \ \ \mbox{on} \ \mbox{the} \ \mbox{space} \ \mathcal{D}[0,T].
$$

From (\ref{2.4}) and and the definition of $\alpha_n$, we have
$$
M^{(n)}(k)=S^{(n)}(k)(1+\xi_k\alpha_n(k)n^{-1/2})
=S^{(n)}(k-1)\exp(\sigma\xi_k n^{-1/2})(1+\xi_k\alpha_n(k)n^{-1/2}).
$$
Then, by
Taylor expansion there exists a constant $c_2$ such that for any $1\leq j\leq n$
$$
\left|\frac{M^{(n)}(j)-M^{(n)}(j-1)}{S^{(n)}(j-1)}-\frac{1}{\sqrt n}
\left((\sigma+\alpha_n(j))\xi_j-\alpha_n(j-1)\xi_{j-1}\right)-
\frac{\sigma}{2n}\left(\sigma + 2 \alpha_n(j) \right)
\right|\leq \frac{c_2}{n^{3/2}}, \ \ \mbox{a.s.}
$$
This together with (\ref{3.9})
yields
that
for any $n\in\mathbb{N}$ and
$t\in [0,1]$
$$
\bigg|Y_n(t)-\frac{\sigma}{\sqrt n}
\sum_{j=1}^{[nt]}\xi_j-
\frac{\sigma }{2n}\big(\sigma [nt]+2\sum_{j=1}^{[nt]}\alpha_n(j)\big)
\bigg|\leq \frac{c_3}{\sqrt n}, \ \mbox{a.s.}\
$$
for some constant $c_3$.
Since $\frac{\sigma}{\sqrt n}\sum_{j=1}^k\xi_j=\ln(S^{(n)}(k)/s_0)$
the above calculations imply that
$$
\int_{0}^t\alpha_{n}([nu])du\ \rightarrow\  \frac{1}{\sigma}\left(Y(t)-
\ln (M(t)/s_0)-\frac{\sigma^2t}{2}\right) = \frac{\langle Y\rangle(t)-\sigma^2t}{2\sigma},\qquad
\tilde{P}\mbox{-a.s.}
$$
\end{proof}

Next, let $c>0$ be a constant and let $M$ be a
strictly positive, continuous martingale
defined on some probability space $(\tilde\Omega,\tilde{\mathcal{F}},\tilde P)$
satisfying the following conditions
\begin{equation}\label{3.20-}
M(0)=s_0 \ \ \mbox{and} \ \
\left|a(t:M)\right| \le c \ \ \mathcal{L}\otimes\tilde{P} \ \ \mbox{a.s.}
\end{equation}
In fact, a volatility process $\alpha \in \mathcal{A}^c$ if and only
if the corresponding process $S_\alpha$ satisfies the
above condition.  However, $S_\alpha$
is defined on the canonical
space $(\Omega_W, \mathcal{F}^{W}, \mathbb{P}^{W})$
and $M$ is defined on a general
space.
In the next lemma,
we show that maximization of the
function $J(M)$ defined in (\ref{3.17++})
over all martingale $M$'s
satisfying the constraint (\ref{3.20-})
is the same as maximizing $J(S_\alpha)$
over $\alpha \in \mathcal{A}^c$.
The proof follows the ideas of
Lemma 5.2 in \cite{K} and uses the
randomization technique.

\begin{lem}
\label{randomization}
Let $M$ be a strictly positive, continuous martingale on
$(\tilde\Omega,\tilde{\mathcal{F}},\tilde P)$
satisfying
{\rm{(}}\ref{3.20-}{\rm{)}}.  Then,
$$
J(M)\leq\sup_{\alpha\in\mathcal{A}^c}J(S_{\alpha}).
$$
\end{lem}
\begin{proof}
Set
$$
Y(t)=\int_{0}^t \frac{dM(u)}{M(u)},\qquad
t\in[0,1],
$$
so that
$$
M(t)=s_0\exp\left(Y(t)-\frac{\langle Y\rangle(t)}{2}\right),
\quad t \in [0,1].
$$
If necessary, by enlarging the space,
we may assume that the probability space
$(\tilde{\Omega},\tilde{F},\tilde{P})$
is rich enough to contain
a Brownian motion $\hat{W}(t)$ which is
independent of $M$. For $\lambda\in [0,1]$
define
$$
Y_{\lambda}=\sqrt{1-\lambda}Y+\sigma\sqrt{\lambda}\hat{W} \ \ \mbox{and} \ \
M_{\lambda}=s_0\exp\left(Y_\lambda-\frac{\langle Y_{\lambda}\rangle}{2}\right).
$$
Notice that for all $\lambda$,
$M_{\lambda}$ satisfies
the conditions of (\ref{3.20-}).
Hence, the family
$$
F(M_{\lambda})-
\int_{0}^1 \widehat{G}\left(t, M_{\lambda},
a(t:M_\lambda) M_{\lambda}(t)\right) dt,\qquad \lambda\in [0,1],
$$
is uniformly integrable, and the continuity of $\widehat{G}$
implies that
$$
J(M)=\lim_{\lambda\rightarrow 0}J(M_{\lambda}).
$$
Hence it suffices to show that
$$
J(M_\lambda)\leq\sup_{\alpha\in\mathcal{A}^c}J(S_{\alpha}),
$$
for all $\lambda>0$.
Since
 $d\langle Y\rangle(t)\geq \lambda\sigma^2 dt$
  for any $\lambda>0$,
without loss of generality we may assume that
$$
Z(t):=\frac{d\langle Y\rangle}{dt}\geq \epsilon,\qquad
\mathcal{L}\otimes\tilde{P}\mbox{-a.s.}
$$
for some $\epsilon>0$. Set,
\begin{eqnarray}\label{3.26}
\tilde{W}(t)&=&\int_{0}^t \frac{d Y(u)}{\sqrt Z(u)}, \ \ t\in [0,1],\\
\kappa_n(0)& =& \sigma \quad \ \mbox{and} \quad
\kappa_n(k)=n\int_{(k-1)/n}^{k/n}\sqrt{Z(u)}du \quad \mbox{for} \quad 0<k<n,
\nonumber\\
M^{(n)}(t)&=& s_0\exp\left(\int_{0}^t\kappa_n([nu])d\tilde{W}(u)-
\frac{1}{2}\int_{0}^t \kappa^2_n([nu])du\right), \ \ t\in[0,1], \ n\in\mathbb{N}.\nonumber
\end{eqnarray}
By the Levy's theorem, $\tilde{W}$ is
a Brownian motion with respect to the usual filtration of $M$.
Therefore,
the martingale $M^{(n)}$ satisfies  (\ref{3.20-}).
Also, from (\ref{3.26}) it is clear that
$$
\lim_{n\rightarrow\infty}\kappa_n([nt])=\sqrt {Z(t})
$$
in probability with the measure $\mathcal{L}\otimes\tilde{P}$.
On the other hand,
Ito's isometry and the Doob-Kolmogorov inequality, imply that
$$
\lim_{n\rightarrow\infty}\max_{0\leq t\leq 1}\left|\int_{0}^t\kappa_n([nu])d\tilde{W}(u)-
Y(t)\right|=0
$$
in probability with respect to $\tilde{P}$.
We use these convergence results and
the uniform integrability, to conclude that
$$
J(M)=\lim_{n\rightarrow\infty}J(M^{(n)}).
$$
Hence, it suffices to prove the following
for any $n\in\mathbb{N}$
\begin{eqnarray}\label{3.29}
J(M^{(n)})\leq\sup_{\alpha\in \mathcal{A}^C}J(S_{\alpha}).
\end{eqnarray}

We  prove the above inequality by the randomization technique.
Fix $n\in\mathbb{N}$. From the existence of the regular distribution function
(for details see \cite{S} page 227),
for any $1\leq k<n$
there exists a function
$\rho_k:\mathbb{R}\times\mathcal{C}[0,1]\times \mathbb{R}^{k}\rightarrow [0,1]$
such that for any
$y$, $\rho_k(y,\cdot):\mathcal{C}[0,1]\times \mathbb{R}^{k}\rightarrow [0,1]$
is measurable and satisfies
$$
\tilde{E}\big(\kappa_n(k)\leq y\big|\sigma\{\tilde{W},\kappa_n(0),...,\kappa_n(k-1)\}\big)
=\rho_k(y,\tilde{W},\kappa_n(0),...,\kappa_n(k-1)),
\ \ \tilde{P} \ \ \mbox{a.s.}
$$
Furthermore, $\tilde{P}$ almost surely,
$\rho_k(\cdot,\tilde{W},\kappa_n(0),...,\kappa_n(k-1))$
is a distribution function on $\mathbb{R}$.
Let $W$ be the Brownian motion in our
canonical space
$(\Omega_W,\mathcal{F}^W,P^{W})$.
We extend this space so that
 it contains a sequence $\Xi_1,...,\Xi_{n-1}$ of i.i.d.~random variables
which are uniformly distributed on the interval $(0,1)$ and independent of
$W$. Let $(\tilde{\Omega}_W,\tilde{F}^W,\tilde{P}^W$)
be the extended probability space.  We assume that its complete.

Next, we recursively define the random variables
\begin{equation}\label{3.31}
U_0=\sigma  \ \ \mbox{and} \  \mbox{for} \  1\leq k< n \ \
U_k=\sup\{y|\rho_k(y,{W},U_1,...,U_{k-1})<\Xi_k\}.
\end{equation}
In view of the properties of the functions $\rho_i$,
we can show that $U_1,...,U_{n-1}$ are measurable.
Furthermore $U_i$ is independent of
$\Xi_{k}$ for any $i<k$. This property together with (\ref{3.31}) yields
that for any $y\in\mathbb{R}$ and $1\leq k<n$,
\begin{eqnarray*}
\tilde{P}^W\big(U_k\leq y\big|\sigma\{W, U_0,...,U_{k-1}\}\big)&=&
\tilde{P}^W\big(\rho_k(y,W,U_0,...,U_{k-1})\geq\Xi_k\big|\sigma\{W, U_0,...,U_{k-1}\}\big)\\
&=& \rho_k(y,W,U_0,...,U_{k-1}).
\end{eqnarray*}
Thus we conclude that
$(W,U_0,...,U_{n-1})$ has the same distribution
as  $(\tilde{W},\kappa_n(0),...,\kappa_n(n-1))$.
Also note that for any $k$
and $t \ge k/n$, $\kappa_n(k)$ is independent of
$\left(\tilde{W}(t)-\tilde{W}(k/n)\right)$.
Furthermore, since for any $k$, $\kappa_n(k)$ takes on values in
the interval $[\sqrt{0\vee \sigma(\sigma-2c)},\sqrt{\sigma(\sigma+2c)}]$,
for $1\leq k<n$ there exist functions
$$
\Theta_k:\mathcal{C}[0,k/n]\times(0,1)^k\rightarrow
[\sqrt{0\vee \sigma(\sigma-2c)},\sqrt{\sigma(\sigma+2c)}],
$$
satisfying
$$
U_k=\Theta_k(W,\Xi_1,...,\Xi_k), \ \ 1\leq k<n
$$
where in the expression above we consider the restriction of $W$ to the interval $[0,k/n]$.
Next we introduce the martingale
$$
S_U(t):=s_0\exp\left(\sum_{i=0}^{[nt]} \left(U_i\left(W\left(\frac{i+1}{n}\right)-
W\left(\frac{i}{n}\right)\right)-\frac{U^2_n(i)}{2n}\right)\right), \ \ t\in [0,1].
$$
Finally, for any $z:=(z_1,...,z_{n-1})\in (0,1)^{n-1}$ define
a stochastic process by
$$
U^{(z)}(t)=\sigma  \  \mbox{if} \  t=0 \  \mbox{and} \ \
U^{(z)}(t)=\Theta_{[nt]}(W,z_1,...,z_{[nt]})   \ \mbox{for}  \ t\in (0,1].
$$
Observe that for any $z\in (0,1)^{n-1}$,
the stochastic process $U^{(z)}\in \mathcal{A}^c$.
We now use the Fubini's theorem to conclude that
\begin{equation}\label{3.35}
J(M^{(n)})=J(S)=\int_{z\in (0,1)^n}J(S_{U^{(z)}})dz_1...dz_n \leq
\sup_{\alpha\in \mathcal{A}^C}J(S_{\alpha})
\end{equation}
and (\ref{3.29}) follows.
\end{proof}

Our final result
is the density of the subset $\mathcal{L}(c)$
defined in Definition \ref{d.vol}
in $\mathcal{A}^{c}$.
The following result is proved by using standard density arguments.
Since we could not find a direct reference we provide a self contained proof.
\begin{lem}\label{density}
For any $c>0$,
$$
\sup_{\alpha\in\mathcal{A}^c}J(S_{\alpha})=
\sup_{\tilde\alpha\in\mathcal{L}(c)}J(S_{\tilde\alpha}).
$$
\end{lem}
\begin{proof}
Let ${\{\phi_n\}}_{n=1}^\infty\subset\mathcal{L}(c)$ be a sequence
which converge in probability (with respect to $\mathcal{L}\otimes P^W$) to some
$\alpha\in\mathcal{A}^C$. By the Ito's isometry and the Doob-Kolmogorov inequality, we
directly conclude that
$S_{\phi_n}$ converges to $S_{\alpha}$
in probability on the space $\mathcal{C}[0,1]$. Then, invoking the uniform integrability
once again, we obtain
$\lim_{n\rightarrow\infty}J (S_{\phi_n})=J(S_{\alpha})$.

Therefore to prove the lemma,
for any $\alpha\in\mathcal{A}^c$ we need
to construct  a sequence ${\{\phi_n\}}_{n=1}^\infty\subset\mathcal{L}(c)$
which converges in probability to $\alpha$.
Moreover, by the decomposition $\alpha=\alpha^{+}-\alpha^{-}$,
without loss of generality, we may assume that $\alpha$ is a non negative stochastic process.
Thus, let $\alpha\in\mathcal{A}^c$ be a non negative stochastic process and let $\delta>0$.
It is well known (see \cite{KS})
that there exists a continuous processes $\phi$
adapted to the Brownian filtration, satisfying
\begin{equation}\label{3.44}
\mathcal{L}\otimes P^W\{|\alpha-\phi|>\delta\}<\delta.
\end{equation}
Since the process $\phi$ is continuous,
for all sufficiently large $m$
\begin{equation}\label{3.45}
P^W\big\{\big(\max_{0\leq k\leq m-2}\sup_{k/m\leq t
\leq (k+2)/m}|\phi(t)-\phi(k/m)|\big)>\delta\big\}<\delta.
\end{equation}
Clearly, for any $1\leq k\leq m$ there exists
a measurable function $\Theta_k:\mathcal{C}[0,k/m]\rightarrow\mathbb{R}$
for which
$$
\theta_k(W)=\phi(k/m), \ \ 1\leq k\leq m
$$
where in the expression above we consider the
restriction of $W$ to the interval $[0,k/m]$.
Fix $k$. It is well known (see for instance \cite{B}, Chapter 1)
that we can find a sequence of bounded Lipschitz continuous functions
$\vartheta_n:\mathcal{C}[0,k/m]\rightarrow\mathbb{R}$, $n\in\mathbb{N}$
such that $\lim_{n\rightarrow\infty}\vartheta_n$
$=\theta_k$
a.s. with respect to the Wiener measure on the space
$\mathcal{C}[0,k/m]$.
We conclude that there exists a constant $\mathcal{H}>0$
and a sequence of functions
$\Theta_k:\mathcal{C}[0,1]\rightarrow\mathbb{R}$,
$1\leq k\leq {m-3}$
such that
for any $z_1,z_2\in \mathcal{C}[0,1]$ and $1\leq k\leq {m-3}$
\begin{eqnarray}
\nonumber
&\mbox{i}.& \ \ \Theta_k(z_1)=\Theta_k(z_2) \ \mbox{if}
 \ z_1(s)=z_2(s) \ \ \mbox{for} \  \mbox{any} \ s\leq k/m,\\
 \label{3.47++}
&\mbox{ii}.& \ \ |\Theta_k(z_1)|\leq \mathcal{H},\\
\label{3.47+++}
&\mbox{iii}. &\ \ |\Theta_k(z_1)-\Theta_k(z_2)|\leq \mathcal{H}\big(||z_1-z_2||\big),\\
\label{3.48}
&\mbox{iv}.& \ \ P^W\big\{|\Theta_k(W)-\phi(k/m)|>\delta\big\}<\delta/m.
\end{eqnarray}
Let $\Theta_{-1},\Theta_0,\Theta_{m-2}:\mathcal{C}[0,1]\rightarrow\mathbb{R}$
be given by $\Theta_{-1}=\Theta_0\equiv \phi(0)$ and $\Theta_{m-2}\equiv\sigma$.
Define $f_1:[0,1]\times\mathcal{C}[0,1]\rightarrow\mathbb{R}$
by
$$
f_1(t,z) =
\left\{
\begin{array}{ll}
([mt]+1-mt)\Theta_{[mt]-1}(z)+(mt-[mt])\Theta_{[mt]}(z) ,
\qquad &{\mbox{if}}\ t < 1-1/m,\\
\sigma, &{\mbox{else.}}
\end{array}
\right.
$$
Denote $a=\sqrt{0\vee \sigma(\sigma-2c)}$ and $b=\sqrt{\sigma(\sigma+2c)}$.
Without loss of generality we assume that $\delta<\min(\sigma-a,b-\sigma)$.
Set,
\begin{equation*}
f(t,z)=((a+\delta)\vee f_1(t,z))\wedge (b-\delta), \ \ t\in [0,1], \ z\in\mathcal{C}[0,1].
\end{equation*}
Using (\ref{3.47++})--(\ref{3.47+++}), we conclude that
for any $0\leq k\leq {m-2}$, $t_1,t_2\in [k/m,(k+1)/m]$ and $z_1,z_2\in\mathcal{C}[0,1]$,
\begin{eqnarray*}
|f(t_2,z_2)-f(t_1,z_1)|&\leq& |f_1(t_2,z_2)-f_1(t_1,z_2)|+|f_1(t_1,z_2)-f_1(t_1,z_1)|\\
&\leq& m|t_1-t_2|(|\Theta_{k-1}(z_2)|+|\Theta_k(z_2)|)+|\Theta_{k-1}(z_2)-\Theta_{k-1}(z_1)|\\
&&+|\Theta_{k}(z_2)-\Theta_{k}(z_1)|\leq 2(\mathcal{H}+\sigma)(m+1)(|t_1-t_2|+||z_1-z_2||).
\end{eqnarray*}
Define the process ${\{\Theta(t)\}}_{t=0}^1$ by
$\Theta(t)=f(t,W)$, $t\in [0,1]$.
By the choice of $\delta$, it follows that
$\Theta\in\mathcal{L}(c)$.
Next, observe that for any $t\in [1/m,1-1/m]$ we have
\begin{equation*}
|\Theta(t)-\phi(t)|\leq \max \left(|\phi(t)-\Theta_{[mt]}(W)|,|\phi(t)-\Theta_{[mt]-1}(W)|\right).
\end{equation*}
Thus for any $t\in [1/m,1-1/m]$
\begin{eqnarray}\label{3.51++}
&|\Theta(t)-\alpha(t)|\leq \left(\max_{0\leq k\leq m-3}\sup_{k/m\leq t\leq(k+2)/m}|\phi(t)-\phi(k/m)|\right)\\
&+\left(\max_{0\leq k\leq m-3}|\phi({k/m})-\Theta_k(W)|\right)+|\alpha(t)-\phi(t)|.\nonumber
\end{eqnarray}
Finally, by combining (\ref{3.44})--(\ref{3.45}), (\ref{3.48}) and (\ref{3.51++}) we get
$$
\mathcal{L}\otimes P^W\{|\Theta-\alpha|>3\delta\}<\frac{2}{m}+3\delta<5\delta.
$$
Since $\delta>0$ was arbitrary small we complete the proof.
\end{proof}

\bibliographystyle{spbasic}

\end{document}